%% file: main.tex
\documentclass[runningheads]{llncs}

\input{preamble}
\newcommand\anonymize

\title{Synthesis of Distributed Protocols by Enumeration Modulo Isomorphisms}

\author{Derek Egolf\textsuperscript{\tiny(\Letter)}
\and
Stavros Tripakis
}
\authorrunning{D. Egolf \& S. Tripakis}
\institute{
Northeastern University, Boston, MA, USA\\
\email{\{egolf.d, stavros\}@northeastern.edu}}
\date{\today}

\begin{document}
\maketitle

\begin{abstract}
Synthesis of distributed protocols is a hard, often undecidable, problem.
{\em Completion} techniques
provide partial remedy by turning the problem into a search problem.
However, the space of candidate completions is still massive.
In this paper, we propose  optimization techniques to reduce the size of the search space by a factorial factor by exploiting symmetries
({\em isomorphisms})
in functionally equivalent solutions. We present both a theoretical analysis of this optimization as well as empirical results that demonstrate its effectiveness in synthesizing
both the Alternating Bit Protocol and Two Phase Commit.
Our experiments show that the optimized tool achieves a speedup of approximately 2 to 10 times compared to its unoptimized counterpart.
\end{abstract}

\input{sections/intronew}

\input{sections/background}
\input{sections/gcg}
\input{sections/opt}

\input{sections/experiments}

\input{sections/related}

\input{sections/end}

\vspace{1.1em}
\subsubsection{Acknowledgements}
Derek Egolf's research has been initially supported by a Northeastern University PhD fellowship.
This material is based upon work supported by the National Science Foundation Graduate Research Fellowship under Grant No. (1938052).
Any opinion, findings, and conclusions or recommendations expressed in this material are those of the authors(s) and do not necessarily reflect the views of the National Science Foundation. We thank Christos Stergiou for his work on the distributed protocol completion tool that we built upon.

%
%
%
\bibliographystyle{plain}
\bibliography{biblio}
%

\input{sections/appendix.tex}

\end{document}

%% file: preamble.tex
\usepackage{amssymb}
\usepackage{mathtools}
\usepackage{mathrsfs}
\usepackage{tikz}
\usetikzlibrary{decorations.pathreplacing,calligraphy}
\usetikzlibrary{patterns}
\usepackage{enumerate}
\usepackage[section]{placeins}
\usepackage[linesnumbered,ruled,vlined]{algorithm2e}
\usepackage{float}
\floatstyle{plaintop}
\restylefloat{table}
\usepackage{fitbox}
\usepackage{longtable}
\usepackage{xfrac}
\usepackage{graphicx}
\usepackage[misc]{ifsym}

\newcommand\tiksize{0.12}

\newcommand{\lt}[1]{\stackrel{#1}{\rightarrow}}

\makeatletter
\def\BState{\State\hskip-\ALG@thistlm}
\makeatother
\DeclarePairedDelimiter\abs{\lvert}{\rvert}

\DeclarePairedDelimiter\angles{\langle}{\rangle}
\DeclarePairedDelimiter\sem{[\![}{]\!]}

\let\phi\varphi

\let\rho\varrho
\let\sat\vDash
\let\unsat\nvDash

\makeatletter
\let\origsubsubsection\subsubsection
\renewcommand\subsubsection{\@ifstar{\starsubsubsection}{\nostarsubsubsection}}

\newcommand\nostarsubsubsection[1]
{\subsubsectionprelude\origsubsubsection{#1}\subsubsectionpostlude}

\newcommand\starsubsubsection[1]
{\subsubsectionprelude\origsubsubsection*{#1}\subsubsectionpostlude}

\newcommand\subsubsectionprelude{%
  \vspace{-1.1em}
}

\newcommand\subsubsectionpostlude{%
  \vspace{0pt}
}

\let\subsubsection\starsubsubsection

\makeatother

\newcommand{\algo}{\textsc{gcg}}
\newcommand{\algov}[3]{\algo\ensuremath{[#1,#2,#3]}}
\newcommand{\algopr}{\algo\ensuremath{_\simeq}}
\newcommand{\algovpr}[3]{\algopr\ensuremath{[#1,#2,#3]}}
\newcommand{\asgn}{\textsc{asgn}}

\newcommand{\SATnv}{\textsc{sat}}
\newcommand{\SAT}[1]{\ensuremath{\SATnv(#1)}}

\newcommand{\MC}{\textsc{mc}}

\newcommand{\gen}{\gamma}
\newcommand{\perm}{\pi}

\newcommand{\eqcl}[1]{\ensuremath{[#1]}}

\newcommand{\trans}[3]{#1\lt{#2}#3}
\newcommand{\bvar}[4]{#2\looparrowright^{#3}_{#1}#4}

\newcommand{\Sol}{{\mathsf{Sol}}}
\newcommand{\Cand}{{\mathsf{Cand}}}
\newcommand{\Pruned}{{\mathsf{Pruned}}}
\newcommand{\pru}[1]{\ensuremath{#1\in\Pruned}}
\newcommand{\notpru}[1]{\ensuremath{#1\notin\Pruned}}
\newcommand{\sol}[1]{\ensuremath{#1\in\Sol}}
\newcommand{\notsol}[1]{\ensuremath{#1\notin\Sol}}

\PassOptionsToPackage{linktocpage}{hyperref}

%% file: sections/intronew.tex
\section{Introduction}
\label{sec:intro}

Distributed protocols are at the heart of the internet, data centers, cloud services, and other types of infrastructure considered indispensable in a modern society.
Yet distributed protocols are also notoriously difficult to get right, and have therefore been one of the primary application domains of formal verification~\cite{Holzmann91,LamportTLAbook2002,Lynch96,NewcombeAmazon2015,ZaveChord2017}.
An even more attractive proposition is distributed protocol {\em synthesis}: given a formal correctness specification $\psi$, automatically generate a distributed protocol that satisfies $\psi$, i.e., that is {\em correct-by-construction}.

Synthesis is a hard problem in general, suffering, like formal verification, from scalability and similar issues. Moreover, for distributed systems, synthesis is generally
undecidable~\cite{FinkbeinerScheweSTTT13,PnueliRosner90,Thistle2005,TripakisIPL}.
Techniques such as program {\em sketching}~\cite{SolarLezama2006,SolarLezamaSTTT13}
remedy scalability and undecidability concerns essentially by turning the synthesis problem into a {\em completion}
problem~\cite{ScenariosHVC2014,CompletionCAV2015}:
given an {\em incomplete} system $M_0$ and a specification $\psi$, automatically synthesize a completion $M$ of $M_0$, such that $M$ satisfies $\psi$.

For example, the synthesis of the well-known {\em alternating-bit protocol} (ABP) is considered in~\cite{sigact2017} as a completion problem:
given an ABP system containing the incomplete $\textit{Sender}_0$ and $\textit{Receiver}_0$ processes shown in Fig.~\ref{fig:sigact-fig14},
complete these two processes (by adding but not removing any transitions, and not adding nor removing any states), so that the system satisfies a given set of requirements.

In cases where the space of all possible completions is finite, completion turns synthesis into a decidable problem.\footnote{
We emphasize that no generality is lost in the sense that one can augment the search for correct completions with an outer loop that keeps adding extra {\em empty} states (with no incoming or outgoing transitions), which the inner completion procedure then tries to complete. Thus, we can keep searching for progressively larger systems (in terms of number of states) until a solution is found, if one exists.
}
However, even then, the number of possible completions can be prohibitively large, even for relatively simple protocols.
For instance, as explained in~\cite{sigact2017}, the number of all possible completions in the ABP example is $512^4\cdot 36$, i.e., approximately 2.5 trillion candidate completions.

\begin{figure}
    \begin{flushright}
    \input{figs/sigact-fig15}
    \end{flushright}
    \vspace*{-20pt}
        \input{figs/sigact-fig14}
    \caption{The incomplete ABP Sender and Receiver processes of~\cite{sigact2017}}
    \label{fig:sigact-fig14}
\end{figure}
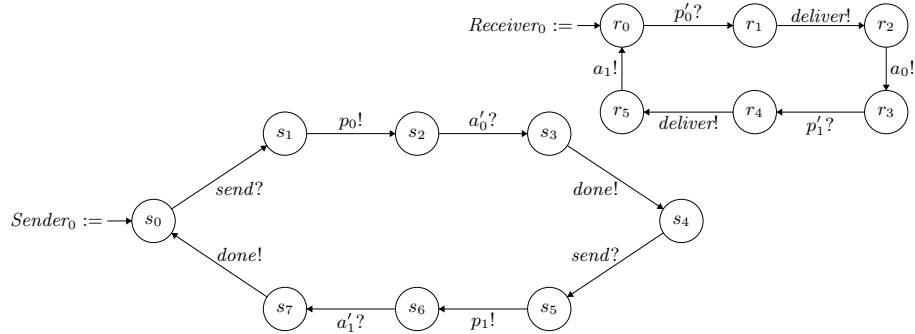

Not only is the number of candidate completions typically huge, but it is often also interesting to generate not just one correct completion, but many.
For instance, suppose both $M_1$ and $M_2$ are (functionally) correct solutions.
We may want to evaluate $M_1$ and $M_2$ also for {\em efficiency} (perhaps using a separate method)~\cite{egolf2022decoupled}.
In general, we may want to synthesize (and then evaluate w.r.t. performance or other metrics) not just one, but in principle {\em all} correct completions.
We call this problem the completion {\em  enumeration} problem, which is the main focus of this paper.

Enumeration is harder than {\em 1-completion} (synthesis of just one correct solution), since the number of correct solutions might be very large.
For instance, in the case of the ABP example described above, the number of correct completions is 16384 and it takes 88 minutes to generate all of them~\cite{sigact2017}.

\begin{figure}
    \input{figs/sndr_tmpl}
    \vspace*{-20pt}
    \caption{An incomplete ABP Sender with permutable states $s_3,s_7$}%
    \label{fig:sndr_tmpl}
\end{figure}
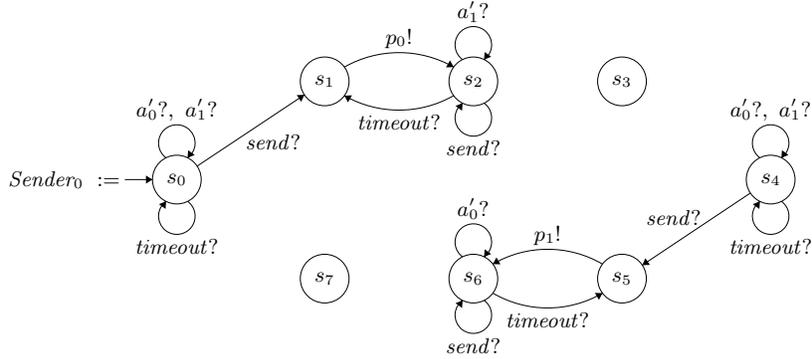

The key idea in this paper is to exploit the notion of {\em isomorphisms} in order to reduce the number of correct completions, as well as the search space of candidate completions in general.
To illustrate the idea, consider a different incomplete $\textit{Sender}_0$ process, shown in Fig.~\ref{fig:sndr_tmpl}.
Two possible completions of this $\textit{Sender}_0$ are shown in Fig.~\ref{fig:sndr_complete}.
Although these two completions are in principle different, they are identical except that states $s_3$ and $s_7$ are swapped.
Our goal is to develop a technique which considers these two completions {\em equivalent up to isomorphism}, and only explores (and returns) one of them.

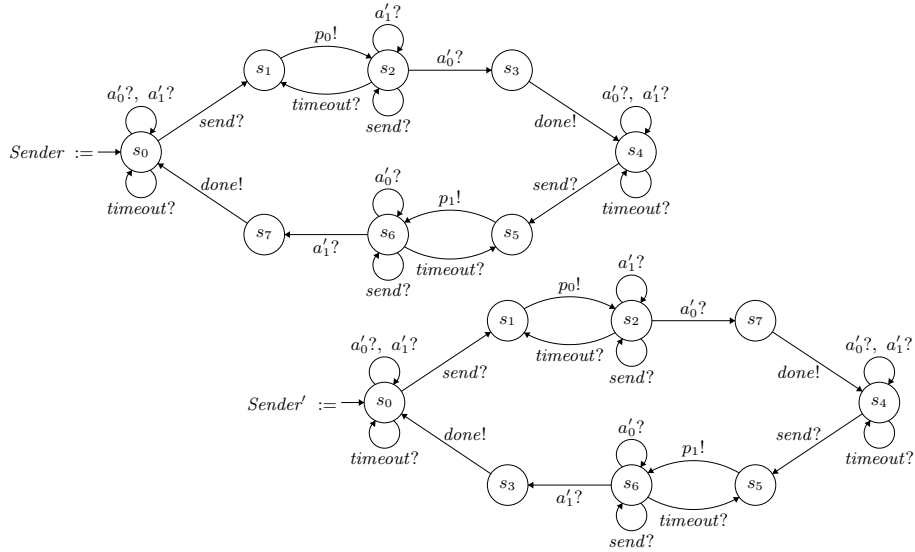
\begin{figure}
    \input{figs/sndr_complete}
    \vspace*{-25pt}
    \begin{flushright}
        \input{figs/sndr_complete2}
    \end{flushright}
    \vspace*{-15pt}
    \caption{Two synthesized completions of the incomplete process of Fig.~\ref{fig:sndr_tmpl}. Observe that the two completions are identical except that states $s_3$ and $s_7$ are flipped.}
    \label{fig:sndr_complete}
\end{figure}

\begin{sloppypar}
To achieve this goal, we adopt the {\em guess-check-generalize} paradigm
(GCG)
    ~\cite{Sygus13,ScenariosHVC2014,GulwaniPolozovSingh2017,SolarLezama2006,SolarLezamaSTTT13}.
In a nutshell, GCG works as follows:
(1) pick a candidate completion $M$;
(2) check
whether $M$ satisfies $\psi$: if it does, $M$ is one possible solution to the synthesis problem;
(3) if $M$ violates $\psi$,
{\em prune} the search space of possible completions by excluding a {\em generalization} of $M$, and repeat from step (1).
In the most trivial case, the generalization of $M$ contains only $M$ itself. Ideally, however, and in order to achieve a more significant pruning of the search space, the generalization of $M$ should contain many more ``bad'' completions which are somehow ``similar'' (for instance, isomorphic) to $M$.
\end{sloppypar}

A naive way to generalize based on isomorphism is
to keep a list of completions encountered thus far and perform an isomorphism check against every element of this list whenever a new candidate is picked. Our approach is smarter: in fact, it does not involve any isomorphism checks whatsoever. Instead, our approach guarantees that no isomorphic completions are ever picked to begin with by pruning them from the search space. This is ultimately done using syntactic transformations of completion representations. The details are left for Section~\ref{sec:opt}.

Furthermore, our notion of ``encountering'' a completion is quite wide. Rather than just pruning completions that are isomorphic to {\em candidates}, we also prune completions that are isomorphic to any completion in the {\em generalizations of} the candidates (with respect to some prior, unextended notion of generalization).
Between the trivial approach
involving isomorphism checks
and our own approach are several other approaches which are good, but not excellent. Indeed, a
categorization
of the subtle differences between such approaches is a key contribution of this paper (see Section~\ref{sec:properties}). These subtleties are easy to miss.

In summary, the main contributions of this paper are the following:
(1) we define the 1-completion and completion-enumeration problems {\em modulo isomorphisms};
(2) we examine new methods to solve these problems based on the GCG paradigm;
(3) we identify properties that an efficient GCG modulo isomorphisms algorithm should have;
(4) we propose two instances of such an algorithm, using a naive and a sophisticated notion of generalization;
(5) we evaluate our methods on the synthesis of two simple distributed protocols: the ABP and Two Phase Commit (2PC) and demonstrate speedups with respect to the unoptimized method of approximately 2 to 10 times.

%% file: figs/sigact-fig15.tex
\scalebox{0.8}{
\begin{tikzpicture}[scale=\tiksize]
\tikzstyle{every node}+=[inner sep=0pt]
\draw [black] (19.3,-32.4) circle (3);
\draw (19.3,-32.4) node {$r_5$};
\draw [black] (37.7,-32.4) circle (3);
\draw (37.7,-32.4) node {$r_4$};
\draw [black] (55.9,-32.4) circle (3);
\draw (55.9,-32.4) node {$r_3$};
\draw [black] (37.7,-20.4) circle (3);
\draw (37.7,-20.4) node {$r_1$};
\draw [black] (19.3,-20.4) circle (3);
\draw (19.3,-20.4) node {$r_0$};
\draw [black] (55.9,-20.4) circle (3);
\draw (55.9,-20.4) node {$r_2$};
\draw [black] (13.1,-20.4) -- (16.3,-20.4);
\draw (12.6,-20.4) node [left] {$\textit{Receiver}_{0}:=$};
\fill [black] (16.3,-20.4) -- (15.5,-19.9) -- (15.5,-20.9);
\draw [black] (22.3,-20.4) -- (34.7,-20.4);
\fill [black] (34.7,-20.4) -- (33.9,-19.9) -- (33.9,-20.9);
\draw (28.5,-19.9) node [above] {$p'_0?$};
\draw [black] (40.7,-20.4) -- (52.9,-20.4);
\fill [black] (52.9,-20.4) -- (52.1,-19.9) -- (52.1,-20.9);
\draw (46.8,-19.9) node [above] {$\textit{deliver}!$};
\draw [black] (55.9,-23.4) -- (55.9,-29.4);
\fill [black] (55.9,-29.4) -- (56.4,-28.6) -- (55.4,-28.6);
\draw (56.4,-26.4) node [right] {$a_0!$};
\draw [black] (52.9,-32.4) -- (40.7,-32.4);
\fill [black] (40.7,-32.4) -- (41.5,-32.9) -- (41.5,-31.9);
\draw (46.8,-32.9) node [below] {$p'_1?$};
\draw [black] (34.7,-32.4) -- (22.3,-32.4);
\fill [black] (22.3,-32.4) -- (23.1,-32.9) -- (23.1,-31.9);
\draw (28.5,-32.9) node [below] {$\textit{deliver}!$};
\draw [black] (19.3,-29.4) -- (19.3,-23.4);
\fill [black] (19.3,-23.4) -- (18.8,-24.2) -- (19.8,-24.2);
\draw (18.8,-26.4) node [left] {$a_1!$};
\end{tikzpicture}
}

%% file: figs/sigact-fig14.tex

\scalebox{0.8}{
\begin{tikzpicture}[scale=\tiksize]
\tikzstyle{every node}+=[inner sep=0pt]
\draw [black] (3.8,-29.2) circle (3);
\draw (3.8,-29.2) node {$s_0$};
\draw [black] (22,-17.1) circle (3);
\draw (22,-17.1) node {$s_1$};
\draw [black] (40.3,-17.1) circle (3);
\draw (40.3,-17.1) node {$s_2$};
\draw [black] (58.6,-17.1) circle (3);
\draw (58.6,-17.1) node {$s_3$};
\draw [black] (76.9,-29.2) circle (3);
\draw (76.9,-29.2) node {$s_4$};
\draw [black] (22,-41.4) circle (3);
\draw (22,-41.4) node {$s_7$};
\draw [black] (40.3,-41.4) circle (3);
\draw (40.3,-41.4) node {$s_6$};
\draw [black] (58.6,-41.4) circle (3);
\draw (58.6,-41.4) node {$s_5$};
\draw [black] (6.3,-27.54) -- (19.5,-18.76);
\fill [black] (19.5,-18.76) -- (18.56,-18.79) -- (19.11,-19.62);
\draw (15.68,-23.65) node [below] {\textit{send}?};
\draw [black] (25,-17.1) -- (37.3,-17.1);
\fill [black] (37.3,-17.1) -- (36.5,-16.6) -- (36.5,-17.6);
\draw (31.15,-16.6) node [above] {$p_0!$};
\draw [black] (43.3,-17.1) -- (55.6,-17.1);
\fill [black] (55.6,-17.1) -- (54.8,-16.6) -- (54.8,-17.6);
\draw (49.45,-16.6) node [above] {$a'_0?$};
\draw [black] (61.1,-18.75) -- (74.4,-27.55);
\fill [black] (74.4,-27.55) -- (74.01,-26.69) -- (73.45,-27.52);
\draw (64.97,-23.65) node [below] {\textit{done}!};
\draw [black] (74.4,-30.86) -- (61.1,-39.74);
\fill [black] (61.1,-39.74) -- (62.04,-39.71) -- (61.48,-38.88);
\draw (64.97,-34.8) node [above] {\textit{send}?};
\draw [black] (55.6,-41.4) -- (43.3,-41.4);
\fill [black] (43.3,-41.4) -- (44.1,-41.9) -- (44.1,-40.9);
\draw (49.45,-41.9) node [below] {$p_1!$};
\draw [black] (37.3,-41.4) -- (25,-41.4);
\fill [black] (25,-41.4) -- (25.8,-41.9) -- (25.8,-40.9);
\draw (31.15,-41.9) node [below] {$a'_1?$};
\draw [black] (19.51,-39.73) -- (6.29,-30.87);
\fill [black] (6.29,-30.87) -- (6.68,-31.73) -- (7.23,-30.9);
\draw (15.68,-34.8) node [above] {\textit{done}!};
\draw [black] (-2.6,-29.2) -- (0.8,-29.2);
\draw (-3.1,-29.2) node [left] {$\textit{Sender}_{0}\mbox{}:=$};
\fill [black] (0.8,-29.2) -- (0,-28.7) -- (0,-29.7);
\end{tikzpicture}
}

%% file: figs/sndr_tmpl.tex

\begin{center}
\scalebox{0.9}{
\begin{tikzpicture}[scale=\tiksize]
\tikzstyle{every node}+=[inner sep=0pt]
\draw [black] (3.8,-29.2) circle (3);
\draw (3.8,-29.2) node {$s_0$};
\draw [black] (22,-17.1) circle (3);
\draw (22,-17.1) node {$s_1$};
\draw [black] (40.3,-17.1) circle (3);
\draw (40.3,-17.1) node {$s_2$};
\draw [black] (58.6,-17.1) circle (3);
\draw (58.6,-17.1) node {$s_3$};
\draw [black] (76.9,-29.2) circle (3);
\draw (76.9,-29.2) node {$s_4$};
\draw [black] (22,-41.4) circle (3);
\draw (22,-41.4) node {$s_7$};
\draw [black] (40.3,-41.4) circle (3);
\draw (40.3,-41.4) node {$s_6$};
\draw [black] (58.6,-41.4) circle (3);
\draw (58.6,-41.4) node {$s_5$};
\draw [black] (2.477,-26.52) arc (234:-54:2.25);
\draw (3.8,-21.95) node [above] {$a'_0?,\mbox{ }a'_1?$};
\fill [black] (5.12,-26.52) -- (6,-26.17) -- (5.19,-25.58);
\draw [black] (6.3,-27.54) -- (19.5,-18.76);
\fill [black] (19.5,-18.76) -- (18.56,-18.79) -- (19.11,-19.62);
\draw (15.68,-23.65) node [below] {\textit{send}?};
\draw [black] (24.423,-15.341) arc (119.66327:60.33673:13.593);
\fill [black] (37.88,-15.34) -- (37.43,-14.51) -- (36.93,-15.38);
\draw (31.15,-13.06) node [above] {$p_0!$};
\draw [black] (74.4,-30.86) -- (61.1,-39.74);
\fill [black] (61.1,-39.74) -- (62.04,-39.71) -- (61.48,-38.88);
\draw (64.97,-34.8) node [above] {\textit{send}?};
\draw [black] (42.7,-39.611) arc (120.28318:59.71682:13.385);
\fill [black] (42.7,-39.61) -- (43.64,-39.64) -- (43.14,-38.78);
\draw (49.45,-37.28) node [above] {$p_1!$};
\draw [black] (37.877,-18.859) arc (-60.33673:-119.66327:13.593);
\fill [black] (24.42,-18.86) -- (24.87,-19.69) -- (25.37,-18.82);
\draw (31.15,-21.14) node [below] {\textit{timeout}?};
\draw [black] (56.2,-43.189) arc (-59.71682:-120.28318:13.385);
\fill [black] (56.2,-43.19) -- (55.26,-43.16) -- (55.76,-44.02);
\draw (49.45,-45.52) node [below] {\textit{timeout}?};
\draw [black] (5.123,-31.88) arc (54:-234:2.25);
\draw (3.8,-36.45) node [below] {\textit{timeout}?};
\fill [black] (2.48,-31.88) -- (1.6,-32.23) -- (2.41,-32.82);
\draw [black] (41.623,-44.08) arc (54:-234:2.25);
\draw (40.3,-48.65) node [below] {\textit{send}?};
\fill [black] (38.98,-44.08) -- (38.1,-44.43) -- (38.91,-45.02);
\draw [black] (38.977,-38.72) arc (234:-54:2.25);
\draw (40.3,-34.15) node [above] {$a'_0?$};
\fill [black] (41.62,-38.72) -- (42.5,-38.37) -- (41.69,-37.78);
\draw [black] (78.223,-31.88) arc (54:-234:2.25);
\draw (76.9,-36.45) node [below] {\textit{timeout}?};
\fill [black] (75.58,-31.88) -- (74.7,-32.23) -- (75.51,-32.82);
\draw [black] (75.577,-26.52) arc (234:-54:2.25);
\draw (76.9,-21.95) node [above] {$a'_0?,\mbox{ }a'_1?$};
\fill [black] (78.22,-26.52) -- (79.1,-26.17) -- (78.29,-25.58);
\draw [black] (38.977,-14.42) arc (234:-54:2.25);
\draw (40.3,-9.85) node [above] {$a'_1?$};
\fill [black] (41.62,-14.42) -- (42.5,-14.07) -- (41.69,-13.48);
\draw [black] (41.623,-19.78) arc (54:-234:2.25);
\draw (40.3,-24.35) node [below] {\textit{send}?};
\fill [black] (38.98,-19.78) -- (38.1,-20.13) -- (38.91,-20.72);
\draw [black] (-2.6,-29.2) -- (0.8,-29.2);
\draw (-3.1,-29.2) node [left] {$\textit{Sender}_0\mbox{ }:=$};
\fill [black] (0.8,-29.2) -- (0,-28.7) -- (0,-29.7);
\end{tikzpicture}
}
\end{center}

%% file: figs/sndr_complete.tex

\scalebox{0.75}{
\begin{tikzpicture}[scale=\tiksize]
\tikzstyle{every node}+=[inner sep=0pt]
\draw [black] (3.8,-29.2) circle (3);
\draw (3.8,-29.2) node {$s_0$};
\draw [black] (22,-17.1) circle (3);
\draw (22,-17.1) node {$s_1$};
\draw [black] (40.3,-17.1) circle (3);
\draw (40.3,-17.1) node {$s_2$};
\draw [black] (58.6,-17.1) circle (3);
\draw (58.6,-17.1) node {$s_3$};
\draw [black] (76.9,-29.2) circle (3);
\draw (76.9,-29.2) node {$s_4$};
\draw [black] (22,-41.4) circle (3);
\draw (22,-41.4) node {$s_7$};
\draw [black] (40.3,-41.4) circle (3);
\draw (40.3,-41.4) node {$s_6$};
\draw [black] (58.6,-41.4) circle (3);
\draw (58.6,-41.4) node {$s_5$};
\draw [black] (2.477,-26.52) arc (234:-54:2.25);
\draw (3.8,-21.95) node [above] {$a'_0?,\mbox{ }a'_1?$};
\fill [black] (5.12,-26.52) -- (6,-26.17) -- (5.19,-25.58);
\draw [black] (6.3,-27.54) -- (19.5,-18.76);
\fill [black] (19.5,-18.76) -- (18.56,-18.79) -- (19.11,-19.62);
\draw (15.68,-23.65) node [below] {\textit{send}?};
\draw [black] (24.423,-15.341) arc (119.66327:60.33673:13.593);
\fill [black] (37.88,-15.34) -- (37.43,-14.51) -- (36.93,-15.38);
\draw (31.15,-13.06) node [above] {$p_0!$};
\draw [black] (43.3,-17.1) -- (55.6,-17.1);
\fill [black] (55.6,-17.1) -- (54.8,-16.6) -- (54.8,-17.6);
\draw (49.45,-16.6) node [above] {$a'_0?$};
\draw [black] (61.1,-18.75) -- (74.4,-27.55);
\fill [black] (74.4,-27.55) -- (74.01,-26.69) -- (73.45,-27.52);
\draw (64.97,-23.65) node [below] {\textit{done}!};
\draw [black] (74.4,-30.86) -- (61.1,-39.74);
\fill [black] (61.1,-39.74) -- (62.04,-39.71) -- (61.48,-38.88);
\draw (64.97,-34.8) node [above] {\textit{send}?};
\draw [black] (42.7,-39.611) arc (120.28318:59.71682:13.385);
\fill [black] (42.7,-39.61) -- (43.64,-39.64) -- (43.14,-38.78);
\draw (49.45,-37.28) node [above] {$p_1!$};
\draw [black] (37.3,-41.4) -- (25,-41.4);
\fill [black] (25,-41.4) -- (25.8,-41.9) -- (25.8,-40.9);
\draw (31.15,-41.9) node [below] {$a'_1?$};
\draw [black] (19.51,-39.73) -- (6.29,-30.87);
\fill [black] (6.29,-30.87) -- (6.68,-31.73) -- (7.23,-30.9);
\draw (15.68,-34.8) node [above] {\textit{done}!};
\draw [black] (37.877,-18.859) arc (-60.33673:-119.66327:13.593);
\fill [black] (24.42,-18.86) -- (24.87,-19.69) -- (25.37,-18.82);
\draw (31.15,-21.14) node [below] {\textit{timeout}?};
\draw [black] (56.2,-43.189) arc (-59.71682:-120.28318:13.385);
\fill [black] (56.2,-43.19) -- (55.26,-43.16) -- (55.76,-44.02);
\draw (49.45,-45.52) node [below] {\textit{timeout}?};
\draw [black] (5.123,-31.88) arc (54:-234:2.25);
\draw (3.8,-36.45) node [below] {\textit{timeout}?};
\fill [black] (2.48,-31.88) -- (1.6,-32.23) -- (2.41,-32.82);
\draw [black] (41.623,-44.08) arc (54:-234:2.25);
\draw (40.3,-48.65) node [below] {\textit{send}?};
\fill [black] (38.98,-44.08) -- (38.1,-44.43) -- (38.91,-45.02);
\draw [black] (38.977,-38.72) arc (234:-54:2.25);
\draw (40.3,-34.15) node [above] {$a'_0?$};
\fill [black] (41.62,-38.72) -- (42.5,-38.37) -- (41.69,-37.78);
\draw [black] (78.223,-31.88) arc (54:-234:2.25);
\draw (76.9,-36.45) node [below] {\textit{timeout}?};
\fill [black] (75.58,-31.88) -- (74.7,-32.23) -- (75.51,-32.82);
\draw [black] (75.577,-26.52) arc (234:-54:2.25);
\draw (76.9,-21.95) node [above] {$a'_0?,\mbox{ }a'_1?$};
\fill [black] (78.22,-26.52) -- (79.1,-26.17) -- (78.29,-25.58);
\draw [black] (38.977,-14.42) arc (234:-54:2.25);
\draw (40.3,-9.85) node [above] {$a'_1?$};
\fill [black] (41.62,-14.42) -- (42.5,-14.07) -- (41.69,-13.48);
\draw [black] (41.623,-19.78) arc (54:-234:2.25);
\draw (40.3,-24.35) node [below] {\textit{send}?};
\fill [black] (38.98,-19.78) -- (38.1,-20.13) -- (38.91,-20.72);
\draw [black] (-2.6,-29.2) -- (0.8,-29.2);
\draw (-3.1,-29.2) node [left] {$\textit{Sender}\mbox{ }:=$};
\fill [black] (0.8,-29.2) -- (0,-28.7) -- (0,-29.7);
\end{tikzpicture}
}

%% file: figs/sndr_complete2.tex

\scalebox{0.75}{
\begin{tikzpicture}[scale=\tiksize]
\tikzstyle{every node}+=[inner sep=0pt]
\draw [black] (3.8,-29.2) circle (3);
\draw (3.8,-29.2) node {$s_0$};
\draw [black] (22,-17.1) circle (3);
\draw (22,-17.1) node {$s_1$};
\draw [black] (40.3,-17.1) circle (3);
\draw (40.3,-17.1) node {$s_2$};
\draw [black] (58.6,-17.1) circle (3);
\draw (58.6,-17.1) node {$s_7$};
\draw [black] (76.9,-29.2) circle (3);
\draw (76.9,-29.2) node {$s_4$};
\draw [black] (22,-41.4) circle (3);
\draw (22,-41.4) node {$s_3$};
\draw [black] (40.3,-41.4) circle (3);
\draw (40.3,-41.4) node {$s_6$};
\draw [black] (58.6,-41.4) circle (3);
\draw (58.6,-41.4) node {$s_5$};
\draw [black] (2.477,-26.52) arc (234:-54:2.25);
\draw (3.8,-21.95) node [above] {$a'_0?,\mbox{ }a'_1?$};
\fill [black] (5.12,-26.52) -- (6,-26.17) -- (5.19,-25.58);
\draw [black] (6.3,-27.54) -- (19.5,-18.76);
\fill [black] (19.5,-18.76) -- (18.56,-18.79) -- (19.11,-19.62);
\draw (15.68,-23.65) node [below] {\textit{send}?};
\draw [black] (24.423,-15.341) arc (119.66327:60.33673:13.593);
\fill [black] (37.88,-15.34) -- (37.43,-14.51) -- (36.93,-15.38);
\draw (31.15,-13.06) node [above] {$p_0!$};
\draw [black] (43.3,-17.1) -- (55.6,-17.1);
\fill [black] (55.6,-17.1) -- (54.8,-16.6) -- (54.8,-17.6);
\draw (49.45,-16.6) node [above] {$a'_0?$};
\draw [black] (61.1,-18.75) -- (74.4,-27.55);
\fill [black] (74.4,-27.55) -- (74.01,-26.69) -- (73.45,-27.52);
\draw (64.97,-23.65) node [below] {\textit{done}!};
\draw [black] (74.4,-30.86) -- (61.1,-39.74);
\fill [black] (61.1,-39.74) -- (62.04,-39.71) -- (61.48,-38.88);
\draw (64.97,-34.8) node [above] {\textit{send}?};
\draw [black] (42.7,-39.611) arc (120.28318:59.71682:13.385);
\fill [black] (42.7,-39.61) -- (43.64,-39.64) -- (43.14,-38.78);
\draw (49.45,-37.28) node [above] {$p_1!$};
\draw [black] (37.3,-41.4) -- (25,-41.4);
\fill [black] (25,-41.4) -- (25.8,-41.9) -- (25.8,-40.9);
\draw (31.15,-41.9) node [below] {$a'_1?$};
\draw [black] (19.51,-39.73) -- (6.29,-30.87);
\fill [black] (6.29,-30.87) -- (6.68,-31.73) -- (7.23,-30.9);
\draw (15.68,-34.8) node [above] {\textit{done}!};
\draw [black] (37.877,-18.859) arc (-60.33673:-119.66327:13.593);
\fill [black] (24.42,-18.86) -- (24.87,-19.69) -- (25.37,-18.82);
\draw (31.15,-21.14) node [below] {\textit{timeout}?};
\draw [black] (56.2,-43.189) arc (-59.71682:-120.28318:13.385);
\fill [black] (56.2,-43.19) -- (55.26,-43.16) -- (55.76,-44.02);
\draw (49.45,-45.52) node [below] {\textit{timeout}?};
\draw [black] (5.123,-31.88) arc (54:-234:2.25);
\draw (3.8,-36.45) node [below] {\textit{timeout}?};
\fill [black] (2.48,-31.88) -- (1.6,-32.23) -- (2.41,-32.82);
\draw [black] (41.623,-44.08) arc (54:-234:2.25);
\draw (40.3,-48.65) node [below] {\textit{send}?};
\fill [black] (38.98,-44.08) -- (38.1,-44.43) -- (38.91,-45.02);
\draw [black] (38.977,-38.72) arc (234:-54:2.25);
\draw (40.3,-34.15) node [above] {$a'_0?$};
\fill [black] (41.62,-38.72) -- (42.5,-38.37) -- (41.69,-37.78);
\draw [black] (78.223,-31.88) arc (54:-234:2.25);
\draw (76.9,-36.45) node [below] {\textit{timeout}?};
\fill [black] (75.58,-31.88) -- (74.7,-32.23) -- (75.51,-32.82);
\draw [black] (75.577,-26.52) arc (234:-54:2.25);
\draw (76.9,-21.95) node [above] {$a'_0?,\mbox{ }a'_1?$};
\fill [black] (78.22,-26.52) -- (79.1,-26.17) -- (78.29,-25.58);
\draw [black] (38.977,-14.42) arc (234:-54:2.25);
\draw (40.3,-9.85) node [above] {$a'_1?$};
\fill [black] (41.62,-14.42) -- (42.5,-14.07) -- (41.69,-13.48);
\draw [black] (41.623,-19.78) arc (54:-234:2.25);
\draw (40.3,-24.35) node [below] {\textit{send}?};
\fill [black] (38.98,-19.78) -- (38.1,-20.13) -- (38.91,-20.72);
\draw [black] (-2.6,-29.2) -- (0.8,-29.2);
\draw (-3.1,-29.2) node [left] {$\textit{Sender}'\mbox{ }:=$};
\fill [black] (0.8,-29.2) -- (0,-28.7) -- (0,-29.7);
\end{tikzpicture}
}

%% file: sections/background.tex
\section{Preliminaries}
\label{sec:background}

\vspace{1.1em}
\subsubsection{Labeled Transition Systems}

A (finite) {\em labeled transition system} (LTS) $M$ is a tuple
$\angles{\Sigma,Q,Q_0,\Delta}$, where
\begin{itemize}
    \item $\Sigma$ is a finite set of transition {\em labels} 
    \item $Q$ is a finite set of {\em states} 
    \item $Q_0\subseteq Q$ is the set of {\em initial states}
    \item $\Delta\subseteq Q\times\Sigma\times Q$ is the {\em transition relation}.
\end{itemize}
We write the transition $(p,a,q)\in\Delta$ as $\trans{p}{a}{q}$.

A {\em run} of $M$ is an infinite sequence $q_0\lt{a_0} q_1\lt{a_1} q_2\lt{a_2}...$,
where $q_0\in Q_0$ and for each $i$ we have $(q_i,a_i,q_{i+1}) \in \Delta$.
The {\em trace} produced by this run is $a_0 a_1 a_2 \cdots$.
Semantically, an LTS $M$ represents a set of infinite traces, denoted $\sem{M}\subseteq\Sigma^\omega$.
Specifically, a trace $a_0 a_1 a_2\cdots$ is in $\sem{M}$ exactly when there exists a run
$q_0\lt{a_0} q_1\lt{a_1} q_2\lt{a_2}...$ of $M$.

\subsubsection{Correctness Specification}
\label{sec_psi}

We will assume that we have some formal notion of {\em specification} and some formal notion of {\em satisfaction} between an LTS $M$ and a specification $\psi$. We write $M\sat\psi$ to denote that $M$ satisfies $\psi$.
Our work is agnostic to what exactly $\psi$ might be (e.g., a temporal logic formula, etc.).

\subsubsection{Completions and Syntactic Constraints}
\label{sec_Phi}

Suppose that $M$ and $M_0$ are two LTSs
with the same set of labels
$\Sigma$, the same set of states $Q$, the same set of initial states $Q_0$,
and with transition relations $\Delta$ and $\Delta_0$, respectively.
We say that
{\em $M$ is a completion of $M_0$}
exactly when $\Delta_0\subseteq \Delta$.
That is, $M$ completes $M_0$ by adding more transitions to it (and not removing any).
For example, each of the two LTSs of Fig.~\ref{fig:sndr_complete} is a completion of the LTS shown in Fig.~\ref{fig:sndr_tmpl}.

Often, we wish to impose some constraints on the kind of synthesized processes that we want to obtain during automated synthesis, other than the global constraints imposed on the system by the correctness specification.
For example, in the formal distributed protocol model proposed in~\cite{sigact2017}, synthesized processes such as the ABP $\textit{Sender}$ and $\textit{Receiver}$ are constrained to satisfy a number of requirements, including absence of deadlocks, determinism of the transition relation, the constraint that each state is either an {\em input state} (i.e., it only receives inputs) or an {\em output state} (i.e., it emits a unique output), the constraint that input states are {\em input-enabled} (i.e., they do not block any inputs), and so on.
Such properties are often syntactic or structural and can be inferred statically by observing the transition relation.
The fact that an LTS is a completion of another LTS can also be captured by such constraints.

Constraints like the above are application-specific, and our approach is agnostic to their precise form and meaning.
We will therefore abstract them away, and assume that there is
a propositional logic formula $\Phi$
which captures the set of all syntactically well-formed candidate completions.
The variable space of $\Phi$ and its precise meaning is application-specific.
We will give a detailed construction of $\Phi$ for LTS in Section~\ref{sec:gcg}.
We write $M\sat\Phi$ when LTS $M$ satisfies the {\em syntactic constraints} $\Phi$.
Let $\sem{\Phi} = \{ M \mid M\sat\Phi\}$.

We say that an LTS is {\em correct} if it satisfies both the syntactic constraints imposed by $\Phi$ and the semantic constraints imposed by $\psi$.

\subsubsection{Computational Problems}

\begin{problem}[Model-Checking]
Given LTS $M$, specification $\psi$, and constraints $\Phi$, check whether $M\sat\psi$ and $M\sat\Phi$.
\end{problem}
    A solution to the model-checking problem is an algorithm, $\MC$, such that for all $M,\Phi,\psi$,
    if $M\sat\Phi$ and $M\sat\psi$ then $\MC(M,\Phi,\psi) = 1$; otherwise, $\MC(M,\Phi,\psi) = 0$.

\begin{problem}[Synthesis]
Given specification $\psi$ and constraints $\Phi$, find, if one exists, LTS $M$ such that $M\sat\psi$ and $M\sat\Phi$.
\end{problem}

\begin{problem}[Completion]
\label{prob_completion}
Given LTS $M_0$, specification $\psi$, and constraints $\Phi$, find, if one exists, a completion $M$ of $M_0$ such that $M\sat\psi$ and $M\sat\Phi$.
\end{problem}

\begin{problem}[Completion enumeration]
\label{prob_completion_enumeration}
Given LTS $M_0$, specification $\psi$, and constraints $\Phi$, find all completions $M$ of $M_0$ such that $M\sat\psi$ and $M\sat\Phi$.
\end{problem}

%% file: sections/gcg.tex
\section{The Guess-Check-Generalize Paradigm}
\label{sec:gcg}

In this section we first propose a generic GCG algorithm and reason about its correctness (Section~\ref{sec_genericGCG}).
We then show how to instantiate this  algorithm to solve Problems~\ref{prob_completion} and~\ref{prob_completion_enumeration} (Section~\ref{sec_concreteGCG}).

\subsection{A Generic GCG Algorithm and its Correctness}
\label{sec_genericGCG}

Algorithm~\ref{fig:newalgo} is a formal description of a generic GCG algorithm.
The algorithm takes as input:
(1) a set of syntactic constraints in the form of a propositional formula $\Phi$, as described in Section~\ref{sec_Phi};
(2) a specification $\psi$ as described in Section~\ref{sec_psi};
and
(3) a {\em generalizer} function $\gamma$, described below.

\begin{algorithm}{}
    \While{$\Phi$ is satisfiable}{
        $\sigma := \SAT{\Phi}$\;\label{lineSAT}
        \uIf{$\MC(M_\sigma,\Phi,\psi)=1$}{\label{lineMC}
          \Return $\sigma$\;\label{lineRETURN}
        $\Phi := \Phi\wedge\neg \sigma$\;\label{lineEXACTLY}
        }
        \Else{
        $\Phi := \Phi\wedge\neg \gen(\sigma)$\;\label{lineGAMMA}
        }
    }
\caption{\algov{\Phi}{\psi}{\gamma}}
\label{fig:newalgo}
\end{algorithm}

$\Phi$ is a propositional logic formula (over a certain set of boolean variables that depends on the application domain at hand)
encoding all possible syntactically valid completions.
Every satisfying assignment $\sigma$ of $\Phi$ corresponds to one completion, which we denote as $M_\sigma$.
Observe that $\algo$
does not explicitly take an initial (incomplete) model $M_0$ as input: this omission is not a problem because $M_0$ can be encoded in $\Phi$, as mentioned in Section~\ref{sec_Phi}.
We explain specifically how to do that in the case of LTS in Section~\ref{sec_concreteGCG}.

The algorithm works as follows:
while $\Phi$ is satisfiable:
Line~\ref{lineSAT}: pick a candidate completion $\sigma$ allowed by $\Phi$ by calling a SAT solver.
Line~\ref{lineMC}: model-check the corresponding model $M_\sigma$ against $\psi$ (by definition, $M_\sigma$ satisfies $\Phi$ because $\sigma$ satisfies $\Phi$).
Line~\ref{lineRETURN}: if $M_\sigma$ satisfies $\psi$ then we have found a correct model: we can return it and terminate if we are solving Problem~\ref{prob_completion}, or return it and continue our search for additional correct models if we are solving Problem~\ref{prob_completion_enumeration}.
In the latter case, in line~\ref{lineEXACTLY} we exclude $\sigma$ from $\Phi$
(slightly abusing notation, we treat $\sigma$ as a formula satisfied exactly and only by $\sigma$, so that $\neg\sigma$ is the formula satisfied by all assignments except $\sigma$).
Line~\ref{lineGAMMA}:
if $M_\sigma$ violates $\psi$, then we exclude from $\Phi$ the {\em generalization} $\gen(\sigma)$ of $\sigma$, and continue our search.

\subsubsection{Generalizers}

A {\em generalizer} is a function $\gen$ which takes an assignment $\sigma$
and returns a propositional logic formula $\gen(\sigma)$ that encodes all ``bad'' assignments that we wish to exclude from $\Phi$.
Ideally, however, $\gen(\sigma)$ will encode many more assignments (and therefore candidate completions), so as to prune as large a part of the search space as possible.
A concrete implementation of $\gen$ may require additional information other than just $\sigma$. For example, $\gen$ may consult the specification $\psi$, counter-examples returned by the model-checker (which are themselves a function of $\psi$ and $\sigma$), and so on. We avoid including all this information in the inputs of $\gen$ to ease presentation. We note that $\psi$ does not change during a run of Algorithm~\ref{fig:newalgo} and therefore $\psi$ can be ``hardwired'' into $\gen$ without loss of generality.

A valid generalizer should include the assignment being generalized and it should only include bad assignments (i.e., it should exclude correct completions).
Formally,
a generalizer $\gen$ is said to be {\em proper} if for all $\sigma$ such that $\sigma\sat\Phi$ and $M_\sigma\unsat\psi$, the following conditions hold:
(1) {\em Self-inclusion}: $\sigma\sat\gen(\sigma)$, and (2) {\em Correct-exclusion}: for any $\rho$, if $\rho\sat\Phi$ and $M_{\rho}\sat\psi$ then $\rho\unsat\gen(\sigma)$.

\subsubsection{The Correctness of GCG}

\begin{lemma}
    \label{termination}
    If $\gen$ is proper then \algov{\Phi}{\psi}{\gen} terminates.
\end{lemma}
\begin{proof}
If $\gen$ is proper then $\gamma(\sigma)$ is guaranteed to include at least $\sigma$.
$\Phi$ is a propositional logic formula, therefore it only has a finite set of satisfying assignments.
Every iteration of the loop removes at least one satisfying assignment from $\Phi$, therefore the algorithm terminates.
\qed
\end{proof}

During a run, Algorithm~\ref{fig:newalgo} returns a (possibly empty) set of assignments $\Sol=\{\sigma_1,\sigma_2,...,\sigma_n\}$, representing the solution to Problems~\ref{prob_completion} or~\ref{prob_completion_enumeration}.
Also during a run, the algorithm guesses candidate assignments by calling the subroutine $\SATnv$ (line~\ref{lineSAT}).
Let
$\Cand$ be the set of all these candidates.
Note that $\Sol\subseteq\Cand$, since every solution returned (line~\ref{lineRETURN}) has been first guessed in line~\ref{lineSAT}.

Whenever the algorithm reassigns $\Phi := \Phi\wedge\neg\phi$, we say that it {\em prunes} $\phi$, i.e., the satisfying assignments of $\phi$ are now excluded from the search. We will need to reason about the set of assignments that have been pruned after a certain {\em partial run} of the program. In such cases we can imagine running the algorithm for some amount of time and pausing it. Then the set $\Pruned$ denotes the set of assignments that have been pruned up until that point. It is true that after the program terminates $\Pruned = \sem{\Phi}\setminus\Cand$, but this equality does not necessarily hold for all partial runs.

\begin{theorem}
\label{sounds-and-complete}
(1)	\algov{\Phi}{\psi}{\gen} is {\em sound}, i.e., for all $\sigma\in\Sol$,
we have $\sigma\sat\Phi$ and $M_\sigma\sat \psi$.
(2) If $\gen$ is proper then
\algov{\Phi}{\psi}{\gen} is {\em complete}, i.e.,
 for all $\sigma\sat\Phi$, if $M_\sigma\sat \psi$ then $\sigma\in\Sol$.
\end{theorem}
\begin{proof}
Every $\sigma\in\Sol$ satisfies $\Phi$ (line~\ref{lineSAT}) and the corresponding $M_\sigma$ satisfies $\psi$ (line~\ref{lineMC}), therefore \algov{\Phi}{\psi}{\gen} is sound.
    Now, suppose that $\gamma$ is proper, and take $\rho$ such that $\rho\sat\Phi$ and $M_\rho\sat\psi$.
	To show completeness, it suffices to show that $\rho\in\Cand$. Then, we also have $\rho\in\Sol$ because $M_\rho$ passes the model-checking test in line~\ref{lineMC}.
	Suppose, for a contradiction, that $\rho\not\in\Cand$, i.e., that $\rho$ is pruned.
	Then there must exist some $\sigma$ such that $\rho\sat\gen(\sigma)$ (line~\ref{lineGAMMA}).
	But $\sigma\sat\Phi$ (line~\ref{lineSAT}), which means that $\rho$ violates the {\em correct-exclusion} property of $\gen$. Contradiction.
	\qed
\end{proof}

\subsection{A Concrete Instance of GCG for LTS}
\label{sec_concreteGCG}

Algorithm~\ref{fig:newalgo} is {\em generic} in the sense that depending on how exactly we instantiate $\Phi$, $\psi$, and $\gen$, we can encode different completion enumeration (and more generally model enumeration) problems, as well as solutions.
We now show how to instantiate Algorithm~\ref{fig:newalgo} to solve Problems~\ref{prob_completion} and~\ref{prob_completion_enumeration} concretely for LTS.

\subsubsection{Encoding LTSs and Completions in Propositional Logic}
Let $M_0 = \angles{\Sigma,Q,Q_0,\Delta_0}$ be an incomplete LTS.
Then we can define a set of boolean variables
$$V := \{\bvar{}{p}{a}{q} \mid p,q\in Q \wedge a\in\Sigma\}$$
so that boolean variable $\bvar{}{p}{a}{q}$ encodes whether transition $p\lt{a}q$ is present or not
(if $p\lt{a}q$ is present, then $\bvar{}{p}{a}{q}$ is true, otherwise it is false).
More formally, let $\asgn_V$ be the set of all assignments over $V$.
An assignment $\sigma\in\asgn_V$ represents LTS $M_\sigma$ with transition relation
$
\Delta_\sigma = \{ (p,a,q) \mid \sigma(\bvar{}{p}{a}{q}) = 1 \}.
$
To enforce $M_\sigma$ to be a completion of $M_0$, we need to enforce that $\Delta_0\subseteq\Delta_\sigma$.
We do so by initializing our syntactic constraints $\Phi$ as $\Phi:=\Phi_{\Delta_0}$, where
$$\Phi_{\Delta_0} := \bigwedge_{\trans{p}{a}{q}\in \Delta_0} \bvar{}{p}{a}{q} .$$
We can then add extra constraints to $\Phi$ such as determinism or absence of deadlocks, as appropriate.

\subsubsection{A Concrete Generalizer for LTS}

Based on the principles of \cite{sigact2017},
we can construct a {\em concrete generalizer} $\gen_\textit{LTS}(\sigma)$ for LTS as
$\gen_\textit{LTS}(\sigma)
:= \gen_\textit{safe}(\sigma) \vee \gen_\textit{live}(\sigma)$,
which we separate into a disjunction of a safety violation generalizer and a liveness violation generalizer.
The safety component $\gen_\textit{safe}$ works on the principle that if LTS $M_\sigma$ violates a safety property, then adding extra transitions will not solve this violation. Thus:
$$\gen_\textit{safe}(\sigma) := \bigwedge_{\{x\in V\mid\sigma(x)=1\}} x .$$

The liveness component $\gen_\textit{live}$ can be defined based on a notion of reachable, ``bad'' cycles that enable something to happen infinitely often. Thus, $\neg\gen_\textit{live}$  captures all LTSs that disable these bad cycles by breaking them or making them unreachable.

It can be shown that the concrete generalizer $\gen_\textit{LTS}$ is proper.
Therefore, the concrete instance \algov{\Phi}{\psi}{\gen_\textit{LTS}} is sound, terminating, and complete, i.e., it solves Problems~\ref{prob_completion} and~\ref{prob_completion_enumeration}.

Even though the concrete generalizer is correct, it is not very effective.
In particular, it does not immediately prune isomorphisms.
There may be $O(n!)$ trivially equivalent completions up to state reordering, where $n$ is the number of states in the LTS.
In the next section we present two optimizations exploiting isomorphisms.

%% file: sections/opt.tex
\section{Synthesis Modulo Isomorphisms}
\label{sec:opt}

\subsection{LTS Isomorphisms}

Intuitively, two LTS are isomorphic if we can rearrange the states of one to obtain the other.
For synthesis purposes, we often wish to
provide as a constraint a set of {\em permutable states} $A$, so as to exclude rearrangements that move states outside of $A$. If we can still rearrange the states of an LTS $M_1$ to obtain another LTS $M_2$ subject to this constraint, then we say that {\em $M_1$ and $M_2$ are isomorphic up to $A$}.
For example, the two LTSs of Fig.~\ref{fig:sndr_complete} are isomorphic up to the set of permutable states $A = \{s_3,s_7\}$.
Strictly speaking, they are permutable up to any set of their states, but we choose $A$ to reflect the fact that those two states have no incoming or outgoing transitions in Fig.~\ref{fig:sndr_tmpl}.
Permuting any other states would yield an LTS that is not a completion of Fig.~\ref{fig:sndr_tmpl}.

We now define isomorphisms formally.
Let $M_0$, $M_1$, and $M_2$ be LTSs with the same $\Sigma,Q,Q_0$, and with transition relations $\Delta_0$, $\Delta_1$,
and $\Delta_2$, respectively.
Suppose that $M_1$ and $M_2$ are both completions of $M_0$.
Let $A\subseteq Q\setminus Q_0$. Then we say $M_1$ and $M_2$ are isomorphic up to $A$,
denoted $M_1\overset{A}{\simeq} M_2$,
if and only if there exists a bijection
$f : A\to A$ (i.e., a {\em permutation}) such that
\begin{equation*}
    \trans{p}{a}{q}\in\Delta_1 \text{ if and only if }
    \trans{f(p)}{a}{f(q)}\in\Delta_2.
\end{equation*}
By default, we will assume that $A$ is the set of non-initial states that have no incoming or outgoing transitions in $M_0$.
In that case we will omit $A$ and write $M_1\simeq M_2$.

\begin{lemma}
LTS isomorphism is an equivalence relation, i.e., it is reflexive, symmetric, and transitive.
\end{lemma}

We use $\eqcl{M}$ to denote the {\em equivalence class} of $M$, i.e., $\eqcl{M}=\{M' \mid M'\simeq M\}$.

\begin{lemma}
\label{lemma_isomorphism_preserves_traces}
    If $M_1\overset{A}{\simeq} M_2$ then $\sem{M_1}=\sem{M_2}$.
\end{lemma}
Lemma~\ref{lemma_isomorphism_preserves_traces} states that LTS isomorphism preserves traces.
More generally, we will assume that our notion of specification is preserved by LTS isomorphism, namely,
that if $M_1\overset{A}{\simeq} M_2$ then for any specification $\psi$, $M_1\sat\psi$ iff $M_2\sat\psi$.

\subsubsection{Isomorphic Assignments}

Two assignments $\sigma$ and $\rho$ are isomorphic if the LTSs that they represent are isomorphic. Hence we write $\sigma\simeq\rho$ if and only if $M_\sigma\simeq M_\rho$.
We write $\eqcl{\rho}$ to denote the equivalence class of $\rho$, i.e., the set of all assignments that are isomorphic to $\rho$. These equivalence classes partition $\Phi$ since $\simeq$ is an equivalence relation.

\subsection{Completion Enumeration Modulo Isomorphisms}

Isomorphisms allow us to focus our attention to Problem~\ref{prob_completion_enumeration_iso} instead of Problem~\ref{prob_completion_enumeration}:

\begin{problem}[Completion enumeration modulo isomorphisms]
\label{prob_completion_enumeration_iso}
Given LTS $M_0$, specification $\psi$, and constraints $\Phi$, find  the set
 $$\{[M]\mid M \mbox{ is a completion of } M_0 \mbox{ such that }M\sat\psi \mbox{ and } M\sat\Phi\}.$$
\end{problem}

Problem~\ref{prob_completion_enumeration_iso} asks that only significantly different (i.e., non-isomorphic) completions are returned to the user.
Problem~\ref{prob_completion_enumeration_iso} can be solved by a simple modification to Algorithm~\ref{fig:newalgo}, namely, to exclude the entire equivalence class $\eqcl{\sigma}$ of any discovered solution $\sigma$, as shown in Algorithm~\ref{fig:newalgo2}, line~\ref{lineEXACTLY-pr}.

\begin{algorithm}
    \While{$\Phi$ is satisfiable}{
        $\sigma := \SAT{\Phi}$\;\label{lineSAT-pr}
        \uIf{$\MC(M_\sigma,\Phi,\psi)=1$}{\label{lineMC-pr}
          \Return $\sigma$\;\label{lineRETURN-pr}
        $\Phi := \Phi\wedge\neg\eqcl{\sigma}$\;\label{lineEXACTLY-pr}
        }
        \Else{
        $\Phi := \Phi\wedge\neg \gen(\sigma)$\;\label{lineGAMMA-pr}
        }
    }
\caption{\algovpr{\Phi}{\psi}{\gamma} solving Problem~\ref{prob_completion_enumeration_iso}}
\label{fig:newalgo2}
\end{algorithm}

\subsection{Properties of an Efficient GCG Algorithm}
\label{sec:properties}
We begin by presenting a list of properties that an efficient instance of GCG ought to satisfy.
Except for Property~\ref{inv:sol}, satisfaction of these properties generally depends on the generalizer used.

\begin{property}
\label{inv:sol}
    For all $\sigma$ that satisfy $\Phi$,
    $\eqcl{\sigma}\cap\Sol$ has 0 or 1 element(s).
    In other words, we return at most one solution per equivalence class.
\end{property}

Property~\ref{inv:sol} asks that only significantly different (i.e., non-isomorphic) completions are returned to the user, thereby solving Problem~\ref{prob_completion_enumeration_iso}, which is our main goal.
In addition, this property implies that the number of completions is kept small, which is important when these are fed as inputs to some other routine (e.g., one that selects a ``highly fit'' completion among all valid completions).

\algopr{}  satisfies Property~\ref{inv:sol}, regardless of the parameters.
However, we can go further, by ensuring that not only we do not return isomorphic completions, but we do not even consider isomorphic candidate completions in the first place:

\begin{property}
\label{inv:cand}
    For all $\sigma$ that satisfy $\Phi$,
    $\eqcl{\sigma}\cap\Cand$ has 0 or 1 element(s).
    In other words, we consider at most one candidate per equivalence class.
\end{property}

Maintaining Property~\ref{inv:cand} now guarantees that we only
call the most expensive subroutines at most once for each equivalence class.
Note that, since $\Sol\subseteq\Cand$, Property~\ref{inv:cand} implies Property~\ref{inv:sol}.

 Property~\ref{inv:cand} is still not entirely satisfactory. For instance,
suppose the algorithm generates $\sigma$ as a candidate and then prunes $\gen(\sigma)$. Now suppose that $\rho\simeq\sigma$. Property~\ref{inv:cand} implies that we {\em cannot} call/prune $\gen(\rho)$.
 Property~\ref{inv:gen} rectifies this:

\begin{property}[invariant]
\label{inv:gen}
    Suppose that \algopr{}
     invokes $\Phi := \Phi \wedge \neg \gen(\sigma)$.
    Then for any $\rho\simeq\sigma$, we should have $\sem{\gen(\rho)}\subseteq\Pruned$.
    In other words,
if we prune $\gen(\sigma)$, we should also prune $\gen(\rho)$ for every $\rho$ isomorphic to $\sigma$.
\end{property}

We note that, contrary to Properties~\ref{inv:sol} and~\ref{inv:cand} which need only hold after termination,
 Property~\ref{inv:gen} is an {\em invariant}: we want it to hold for all {\em partial executions} of the  algorithm.

\begin{theorem}
    Suppose $\gen$ is proper. If \algovpr{\Phi}{\psi}{\gen} maintains Property \ref{inv:gen} as an invariant, then \algovpr{\Phi}{\psi}{\gen} also maintains Property~\ref{inv:cand}.
\end{theorem}

Maintaining Property~\ref{inv:gen} increases the rate at which the search space is pruned,
but is still not enough.
Suppose that $\tau\sat\gen(\sigma)$ and that $\tau'\simeq\tau$. If we prune the members of
$\gen(\sigma)$, then we will prune $\tau$, but not necessarily $\tau'$.
This possibility is unsatisfactory, since $\tau$ and $\tau'$ should both be treated whenever one of them is.

\begin{property}[invariant]
\label{inv:pru}
    Suppose $\pru{\tau}$ and $\tau'\simeq\tau$.
    Then $\pru{\tau'}$ or $\sol{\tau'}$.
    In other words, if we prune $\tau$ we should also prune
	any isomorphic $\tau'$, unless $\tau'$ happens to be a solution.
(Note that Property~\ref{inv:sol} guarantees that this exception applies to at most one $\tau'$).
\end{property}

Maintaining Property~\ref{inv:pru} as an invariant further accelerates pruning.
Under certain conditions, Property~\ref{inv:gen} implies Property~\ref{inv:pru}.
In particular, Property~\ref{inv:gen} implies Property~\ref{inv:pru} if $\gen$ is {\em invertible}, a concept that we define next.

\subsubsection{Invertible Generalizers}

Let $\gen$ be a generalizer and let $\tau$ be an assignment.
We define the {\em inverse} $\gen^{-1}(\tau)$, to be the propositional logic formula satisfied by all $\sigma$ such that $\tau\sat\gen(\sigma)$. That is, $\sigma\sat\gen^{-1}(\tau)$ iff $\tau\sat\gen(\sigma)$.

Let $\phi$ and $\phi'$ be propositional logic formulas.
Suppose that for every $\sigma\sat\phi$, there exists a $\sigma'\sat\phi'$ such that $\sigma'\simeq\sigma$.
 Then we say that {\em$\phi$ subsumes $\phi'$ up to isomorphism}. If $\phi$ and $\phi'$ both subsume each other, then we say that they are {\em equivalent up to isomorphism}.

A generalizer $\gamma$ is {\em invertible} if for all assignments $\tau,\tau'$ that satisfy $\Phi$,
if $\tau\simeq\tau'$ then
$\gen^{-1}(\tau)$ and $\gen^{-1}(\tau')$ are equivalent up to isomorphism.
Now if $\tau\sat\gen(\sigma)$ and $\tau'\simeq\tau$, invertibility guarantees that we can point to a $\sigma'\simeq\sigma$ such that $\tau'\sat\gen(\sigma')$.

\begin{theorem}
    \label{invertible}
    Suppose $\gen$ is proper and invertible. If \algovpr{\Phi}{\psi}{\gen} maintains Property \ref{inv:gen} as an invariant, then \algovpr{\Phi}{\psi}{\gen} also maintains Property~\ref{inv:pru} as an invariant.
\end{theorem}
\input{proofs/invert-pf}

It can be shown that the generalizer $\gen_\textit{LTS}$ is invertible.
Essentially, this is because $\gen_\textit{LTS}$ does not depend on state names
(for example, the structure of cycles and paths is independent of state names).
Still, \algovpr{\Phi}{\psi}{\gen_\textit{LTS}} satisfies only Property~\ref{inv:sol} above.
Therefore, we will next describe an optimized generalization method that exploits isomorphism to satisfy all properties.

\subsection{Optimized Generalization}
\label{subsec_opt}

\vspace{1.1em}
\subsubsection{Equivalence Closure}

If $\gen$ is a generalizer and $\simeq$ is an equivalence relation, then let
$$\overset{\simeq}{\gen}(\rho) := \bigvee_{\sigma\in\eqcl{\rho}}\gen(\sigma)$$
be the {\em equivalence closure} of $\gen$.
If $\gen(\sigma) \equiv\overset{\simeq}{\gen}(\sigma)$ for all $\sigma$, we say that $\gen$ is {\em closed under equivalence}.

Note that $\overset{\simeq}{\gen}$ is itself a generalizer. An instance of \algopr{} that uses $\overset{\simeq}{\gen}$ is correct and satisfies all the efficiency properties identified above:

\begin{theorem}
\label{thm_correctness}
    If $\gen$ is a proper generalizer,
    then \algovpr{\Phi}{\psi}{\overset{\simeq}{\gen}} is sound, terminating, and complete up to isomorphisms.
\end{theorem}
\begin{theorem}
\label{thm_properties123}
    If $\gen$ is proper, then \algovpr{\Phi}{\psi}{\overset{\simeq}{\gen}} maintains Properties~\ref{inv:sol} and~\ref{inv:cand}. Furthermore, the algorithm maintains Property~\ref{inv:gen} as an invariant.
\end{theorem}
\begin{theorem}
\label{thm_property4}
    If $\gen$ is both proper and invertible, then: (1) $\overset{\simeq}{\gen}$ is invertible;
    (2) \algovpr{\Phi}{\psi}{\overset{\simeq}{\gen}} maintains Property~\ref{inv:pru} as an invariant.
\end{theorem}

\subsubsection{Computation Options for $\overset{\simeq}{\gen}$}

The naive way to compute $\overset{\simeq}{\gen}$ is to iterate over all $\sigma_1,\sigma_2,\cdots,\sigma_k\in\eqcl{\rho}$, compute each $\gen(\sigma_i)$, and then return the disjunction of all $\gen(\sigma_i)$. We call this the {\em naive generalization} approach.
The problem with this approach is that we have to call $\gen$ as many as $n!$ times, where $n$ is the number of permutable states. The experimental results in Section~\ref{sec:eval} indicate empirically that this naive method does not scale well.

We thus propose a better approach, which is
{\em incremental},
in the sense that we only have to compute $\gen$ once,
for $\gen(\sigma_1)$; we can then perform simple {\em syntactic transformations} on $\gen(\sigma_1)$ to obtain
$\gen(\sigma_2)$, $\gen(\sigma_3)$, and so on. As we will show, these transformations are much more efficient than computing each $\gen(\sigma_i)$ from scratch.
So-called {\em permuters} formalize this idea:

\subsubsection{Permuters}

A {\em permuter} is a function $\perm$
that takes as input an assignment $\rho$ and the generalization $\gen(\sigma)$ for some $\sigma\simeq\rho$, and returns a propositional logic formula $\perm(\sigma,\gen(\rho))$ such that
$\forall \rho\sat\Phi, \forall \sigma\simeq\rho :: M_\rho\unsat\psi \to \perm(\rho,\gen(\sigma)) \equiv \gen(\rho)$.
That is, assuming $\rho$ is ``bad'' ($M_\rho\unsat\psi$), $\perm(\rho,\gen(\sigma))$ is equivalent to $\gen(\rho)$.
However, contrary to $\gen$, $\perm$ can use the extra information $\gen(\sigma)$ to compute the generalization of $\rho$.
Then, instead of $\overset{\simeq}{\gen}(\rho)$, we can compute the logically equivalent formula
$$\gen_\perm(\rho) := \bigvee_{\sigma\in\eqcl{\rho}}\perm(\sigma,\gen(\rho)).$$

\begin{theorem}
\label{thm_permuter}
Theorems~\ref{thm_correctness},~\ref{thm_properties123}, and~\ref{thm_property4}
also hold for \algovpr{\Phi}{\psi}{\gen_\perm}.
\end{theorem}
\begin{proof}
Follows from the fact that for any $\rho$, $\gen_\perm(\rho)\equiv\overset{\simeq}{\gen}(\rho)$.
\qed
\end{proof}

\subsubsection{A Concrete Permuter for LTS}

We now explain how to compute $\perm$ concretely in our application domain, namely LTS.
Let $M_0$ be an incomplete LTS.
Let $\sigma_1,\sigma_2$ be two assignments encoding completions $M_{\sigma_1}$ and $M_{\sigma_2}$ of $M_0$.
Suppose $M_{\sigma_1} \overset{A}{\simeq} M_{\sigma_2}$. Recall that $A$ is the set of permutable states (the non-initial states with no incoming/outgoing transitions by default).
Then there is a permutation $f : A\to A$, such that applying $f$ to the states of $M_{\sigma_1}$ yields $M_{\sigma_2}$.
$f$ allows us to transform one LTS to another, but it also allows us to transform the generalization formula for $\sigma_1$, namely $\gen(\sigma_1)$, to the one for $\sigma_2$, namely $\gen(\sigma_2)$.

For example, let $M_0$ be the leftmost LTS in Fig.~\ref{fig:opt-example}, with alphabet $\Sigma = \{a\}$, states $Q = \{p_0,p_1,p_2,p_3\}$, initial state $p_0$, and the empty transition relation.
Let $M_{\sigma_1}$ and $M_{\sigma_2}$ be the remaining LTSs shown in Fig.~\ref{fig:opt-example}.
Let $A = \{p_1, p_2, p_3\}$ and let $f$ be the permutation mapping $p_1$ to $p_3$, $p_3$ to $p_2$, and $p_2$ to $p_1$.
Then $M_{\sigma_1} \overset{A}{\simeq} M_{\sigma_2}$ and $f$ is the witness to this isomorphism.

Let $\gen(\sigma_1) = (\bvar{}{p_0}{a}{p_1}) \wedge (\bvar{}{p_1}{a}{p_2}) \wedge (\bvar{}{p_2}{a}{p_3})$.
$\gen(\sigma_1)$ captures the four LTSs in Fig.~\ref{fig:M1}.
The key idea is that we can compute $\gen(\sigma_2)$ by transforming $\gen(\sigma_1)$  {\em purely syntactically}.
In particular, we apply the permutation $f$ to all $p_i$ appearing in the variables of the formula.
Doing so, we obtain $\gen(\sigma_2) = (\bvar{}{p_0}{a}{p_3}) \wedge (\bvar{}{p_3}{a}{p_1}) \wedge (\bvar{}{p_1}{a}{p_2})$. This formula in turn captures the four LTSs in Fig.~\ref{fig:M2}, which are exactly the permutations of those in Fig.~\ref{fig:M1} after applying $f$.

We now describe this transformation formally.
Observe that $M_{\sigma_1}$ and $M_{\sigma_2}$ have the same set of states, say $Q$.
We extend the permutation to $f : Q\to Q$
 by defining $f(q) = q$ for all states $q\notin A$. Now, we extend this permutation of states to permutations of the set $V$
 (the set of boolean variables encoding transitions).
Specifically we extend $f$ to permute $V$ by defining:
$f(\bvar{}{p}{a}{q}) := \bvar{}{f(p)}{a}{f(q)}$
and we extend it to propositional formulas by applying it to all variables in the formula.
Then
we define
$\perm_\textit{LTS}(\sigma_2,\gen(\sigma_1)) := f(\gen(\sigma_1))$.

In essence, the permuter $\perm_\textit{LTS}$ identifies the permutation $f$ witnessing the fact that $\sigma_1\simeq\sigma_2$. It then applies $f$ to the variables of $\gen(\sigma_1)$. Applying $f$ to $\gen(\sigma_1)$ is equivalent to applying $f$ to all assignments that satisfy $\gen(\sigma_1)$.

\begin{figure}
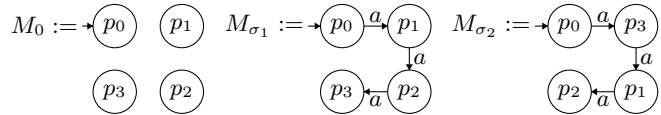

    \vspace*{-15pt}
    \include{figs/opt-example}
    \vspace*{-30pt}
    \caption{An incomplete LTS $M_0$ and two possible completions, $M_{\sigma_1}$ and $M_{\sigma_2}$}
    \label{fig:opt-example}
    \vspace*{-15pt}
\end{figure}
\begin{figure}
    \include{figs/M1}
    \vspace*{-30pt}
    \caption{LTSs represented by $(\bvar{}{p_0}{a}{p_1}) \wedge (\bvar{}{p_1}{a}{p_2}) \wedge (\bvar{}{p_2}{a}{p_3})$}
    \label{fig:M1}
\end{figure}
\begin{figure}
    \include{figs/M2}
    \vspace*{-30pt}
    \caption{LTSs represented by $(\bvar{}{p_0}{a}{p_3}) \wedge (\bvar{}{p_3}{a}{p_1}) \wedge (\bvar{}{p_1}{a}{p_2})$}
    \label{fig:M2}
\end{figure}

It can be shown that $\perm_\textit{LTS}$ is a permuter for LTS.
It follows then that the concrete instance \algovpr{\Phi}{\psi}{\gen_\perm} (where $\gen := \gen_\textit{LTS}$ and $\perm := \perm_\textit{LTS}$)
satisfies
Theorem~\ref{thm_permuter},
i.e., it is sound, terminating, complete up to isomorphisms, and satisfies all Properties~\ref{inv:sol}-\ref{inv:pru}.

%% file: proofs/invert-pf.tex

\begin{proof}

    Let $\gen$ be a proper, invertible generalizer.
    We will proceed by contradiction.
    Assume that we have run the algorithm for some amount of time and paused its execution, freezing the state of $\Pruned$.
    Suppose that \algovpr{\Phi}{\psi}{\gen} satisfies Property~\ref{inv:gen} at this point, but that it does not satisfy Property~\ref{inv:pru}.
    From the negation of Property~\ref{inv:pru}, we have at this point in the execution two assignments $\tau$ and $\tau'$ such that
    (1) $\tau\simeq\tau'$,
    (2) $\pru{\tau}$,
    (3) $\notpru{\tau'}$, and
    (4) $\notsol{\tau'}$

    There are two cases that fall out of (2). Either $\tau$ was pruned using a call to $\gen$, or exactly $\eqcl{\tau}$ was pruned. In the second case, we quickly reach a contradiction since it implies that $\pru{\tau'}$, violating assumption (3).

    So instead, we assume $\tau\sat\gen(\sigma)$
    for some $\sigma$
    and that
	this call to $\gen$ was invoked at some point in the past. So $\sigma\sat\gen^{-1}(\tau)$.
    But then by invertibility and (1) there exists $\sigma'\simeq\sigma$
    such that $\sigma'\sat\gen^{-1}(\tau')$ and hence $\tau'\sat\gen(\sigma')$.
    Property~\ref{inv:gen} tells us then that $\pru{\tau'}$, but this conclusion also violates assumption (3).
	\qed
\end{proof}

%% file: figs/opt-example.tex

\begin{center}
\begin{tikzpicture}[scale=\tiksize*0.8]
\tikzstyle{every node}+=[inner sep=0pt]
\draw [black] (5.6,-33.6) circle (3);
\draw (5.6,-33.6) node {$p_3$};
\draw [black] (14.9,-33.6) circle (3);
\draw (14.9,-33.6) node {$p_2$};
\draw [black] (5.6,-24.5) circle (3);
\draw (5.6,-24.5) node {$p_0$};
\draw [black] (14.9,-24.5) circle (3);
\draw (14.9,-24.5) node {$p_1$};
\draw [black] (37.1,-33.6) circle (3);
\draw (37.1,-33.6) node {$p_3$};
\draw [black] (37.1,-24.5) circle (3);
\draw (37.1,-24.5) node {$p_0$};
\draw [black] (46.4,-33.6) circle (3);
\draw (46.4,-33.6) node {$p_2$};
\draw [black] (46.4,-24.5) circle (3);
\draw (46.4,-24.5) node {$p_1$};
\draw [black] (68.6,-33.6) circle (3);
\draw (68.6,-33.6) node {$p_2$};
\draw [black] (77.8,-33.6) circle (3);
\draw (77.8,-33.6) node {$p_1$};
\draw [black] (68.6,-24.5) circle (3);
\draw (68.6,-24.5) node {$p_0$};
\draw [black] (77.8,-24.5) circle (3);
\draw (77.8,-24.5) node {$p_3$};
\draw [black] (1.1,-24.5) -- (2.6,-24.5);
\draw (0.6,-24.5) node [left] {$M_0 :=$};
\fill [black] (2.6,-24.5) -- (1.8,-24) -- (1.8,-25);

\draw (41.75,-24) node [above] {$a$};
\draw [black] (40.1,-24.5) -- (43.4,-24.5);
\fill [black] (43.4,-24.5) -- (42.6,-24) -- (42.6,-25);

\draw (46.9,-29.05) node [right] {$a$};
\draw [black] (46.4,-27.5) -- (46.4,-30.6);
\fill [black] (46.4,-30.6) -- (46.9,-29.8) -- (45.9,-29.8);

\draw (41.75,-34.1) node [below] {$a$};
\draw [black] (43.4,-33.6) -- (40.1,-33.6);
\fill [black] (40.1,-33.6) -- (40.9,-34.1) -- (40.9,-33.1);

\draw (73.25,-24) node [above] {$a$};
\draw [black] (71.6,-24.5) -- (74.8,-24.5);
\fill [black] (74.8,-24.5) -- (74,-24) -- (74,-25);

\draw (78.3,-29.05) node [right] {$a$};
\draw [black] (77.8,-27.5) -- (77.8,-30.6);
\fill [black] (77.8,-30.6) -- (78.3,-29.8) -- (77.3,-29.8);

\draw (73.25,-34.1) node [below] {$a$};
\draw [black] (74.8,-33.6) -- (71.6,-33.6);
\fill [black] (71.6,-33.6) -- (72.4,-34.1) -- (72.4,-33.1);

\draw [black] (32.4,-24.5) -- (34.1,-24.5);
\draw (31.9,-24.5) node [left] {$M_{\sigma_1}:=$};
\fill [black] (34.1,-24.5) -- (33.3,-24) -- (33.3,-25);
\draw [black] (63.8,-24.5) -- (65.6,-24.5);
\draw (63.3,-24.5) node [left] {$M_{\sigma_2}:=$};
\fill [black] (65.6,-24.5) -- (64.8,-24) -- (64.8,-25);
\end{tikzpicture}
\end{center}

%% file: figs/M1.tex

\begin{center}
\begin{tikzpicture}[scale=\tiksize*0.8]
\tikzstyle{every node}+=[inner sep=0pt]

    \draw [black] (5.6,-33.6) circle (3);
    \draw (5.6,-33.6) node {$p_3$};
    \draw [black] (14.9,-33.6) circle (3);
    \draw (14.9,-33.6) node {$p_2$};
    \draw [black] (5.6,-24.5) circle (3);
    \draw (5.6,-24.5) node {$p_0$};
    \draw [black] (14.9,-24.5) circle (3);
    \draw (14.9,-24.5) node {$p_1$};
    \draw [black] (27.1,-33.6) circle (3);
    \draw (27.1,-33.6) node {$p_3$};
    \draw [black] (27.1,-24.5) circle (3);
    \draw (27.1,-24.5) node {$p_0$};
    \draw [black] (36.3,-33.6) circle (3);
    \draw (36.3,-33.6) node {$p_2$};
    \draw [black] (36.3,-24.5) circle (3);
    \draw (36.3,-24.5) node {$p_1$};
    \draw [black] (48.5,-33.6) circle (3);
    \draw (48.5,-33.6) node {$p_3$};
    \draw [black] (57.7,-33.6) circle (3);
    \draw (57.7,-33.6) node {$p_2$};
    \draw [black] (48.5,-24.5) circle (3);
    \draw (48.5,-24.5) node {$p_0$};
    \draw [black] (57.7,-24.5) circle (3);
    \draw (57.7,-24.5) node {$p_1$};
    \draw [black] (69.8,-33.6) circle (3);
    \draw (69.8,-33.6) node {$p_3$};
    \draw [black] (79,-33.6) circle (3);
    \draw (79,-33.6) node {$p_2$};
    \draw [black] (69.8,-24.5) circle (3);
    \draw (69.8,-24.5) node {$p_0$};
    \draw [black] (79,-24.5) circle (3);
    \draw (79,-24.5) node {$p_1$};
    \draw [black] (1.4,-24.5) -- (2.6,-24.5);
    \fill [black] (2.6,-24.5) -- (1.8,-24) -- (1.8,-25);
    \draw [black] (30.1,-24.5) -- (33.3,-24.5);
    \fill [black] (33.3,-24.5) -- (32.5,-24) -- (32.5,-25);
    \draw (31.7,-24) node [above] {$a$};
    \draw [black] (36.3,-27.5) -- (36.3,-30.6);
    \fill [black] (36.3,-30.6) -- (36.8,-29.8) -- (35.8,-29.8);
    \draw (35.8,-29.05) node [left] {$a$};
    \draw [black] (33.3,-33.6) -- (30.1,-33.6);
    \fill [black] (30.1,-33.6) -- (30.9,-34.1) -- (30.9,-33.1);
    \draw (31.7,-34.1) node [below] {$a$};
    \draw [black] (51.5,-24.5) -- (54.7,-24.5);
    \fill [black] (54.7,-24.5) -- (53.9,-24) -- (53.9,-25);
    \draw (53.1,-24) node [above] {$a$};
    \draw [black] (57.7,-27.5) -- (57.7,-30.6);
    \fill [black] (57.7,-30.6) -- (58.2,-29.8) -- (57.2,-29.8);
    \draw (58.2,-29.05) node [right] {$a$};
    \draw [black] (54.7,-33.6) -- (51.5,-33.6);
    \fill [black] (51.5,-33.6) -- (52.3,-34.1) -- (52.3,-33.1);
    \draw (53.1,-34.1) node [below] {$a$};
    \draw [black] (22.4,-24.5) -- (24.1,-24.5);
    \fill [black] (24.1,-24.5) -- (23.3,-24) -- (23.3,-25);
    \draw [black] (43.7,-24.5) -- (45.5,-24.5);
    \fill [black] (45.5,-24.5) -- (44.7,-24) -- (44.7,-25);
    \draw [black] (8.6,-24.5) -- (11.9,-24.5);
    \fill [black] (11.9,-24.5) -- (11.1,-24) -- (11.1,-25);
    \draw (10.25,-24) node [above] {$a$};
    \draw [black] (14.9,-27.5) -- (14.9,-30.6);
    \fill [black] (14.9,-30.6) -- (15.4,-29.8) -- (14.4,-29.8);
    \draw (15.4,-29.05) node [right] {$a$};
    \draw [black] (11.9,-33.6) -- (8.6,-33.6);
    \fill [black] (8.6,-33.6) -- (9.4,-34.1) -- (9.4,-33.1);
    \draw (10.25,-34.1) node [below] {$a$};
    \draw [black] (5.6,-30.6) -- (5.6,-27.5);
    \fill [black] (5.6,-27.5) -- (5.1,-28.3) -- (6.1,-28.3);
    \draw (5.1,-29.05) node [left] {$a$};
    \draw [black] (29.23,-31.49) -- (34.17,-26.61);
    \fill [black] (34.17,-26.61) -- (33.25,-26.82) -- (33.95,-27.53);
    \draw (32.67,-29.53) node [below] {$a$};
    \draw [black] (50.592,-31.507) arc (118.60039:61.39961:5.239);
    \fill [black] (55.61,-31.51) -- (55.14,-30.69) -- (54.67,-31.56);
    \draw (53.1,-30.37) node [above] {$a$};
    \draw [black] (79,-27.5) -- (79,-30.6);
    \fill [black] (79,-30.6) -- (79.5,-29.8) -- (78.5,-29.8);
    \draw (79.5,-29.05) node [right] {$a$};
    \draw [black] (72.8,-24.5) -- (76,-24.5);
    \fill [black] (76,-24.5) -- (75.2,-24) -- (75.2,-25);
    \draw (74.4,-24) node [above] {$a$};
    \draw [black] (76,-33.6) -- (72.8,-33.6);
    \fill [black] (72.8,-33.6) -- (73.6,-34.1) -- (73.6,-33.1);
    \draw (74.4,-34.1) node [below] {$a$};
    \draw [black] (65.2,-24.5) -- (66.8,-24.5);
    \fill [black] (66.8,-24.5) -- (66,-24) -- (66,-25);
    \draw [black] (70.785,-30.779) arc (188.48384:-99.51616:2.25);
    \draw (75.14,-27.4) node [right] {$a$};
    \fill [black] (72.64,-32.67) -- (73.5,-33.04) -- (73.36,-32.05);

\end{tikzpicture}
\end{center}

%% file: figs/M2.tex

\begin{center}
\begin{tikzpicture}[scale=\tiksize*0.8]
\tikzstyle{every node}+=[inner sep=0pt]

    \draw [black] (5.6,-33.6) circle (3);
    \draw (5.6,-33.6) node {$p_2$};
    \draw [black] (14.9,-33.6) circle (3);
    \draw (14.9,-33.6) node {$p_1$};
    \draw [black] (5.6,-24.5) circle (3);
    \draw (5.6,-24.5) node {$p_0$};
    \draw [black] (14.9,-24.5) circle (3);
    \draw (14.9,-24.5) node {$p_3$};
    \draw [black] (27.1,-33.6) circle (3);
    \draw (27.1,-33.6) node {$p_2$};
    \draw [black] (27.1,-24.5) circle (3);
    \draw (27.1,-24.5) node {$p_0$};
    \draw [black] (36.3,-33.6) circle (3);
    \draw (36.3,-33.6) node {$p_1$};
    \draw [black] (36.3,-24.5) circle (3);
    \draw (36.3,-24.5) node {$p_3$};
    \draw [black] (48.5,-33.6) circle (3);
    \draw (48.5,-33.6) node {$p_2$};
    \draw [black] (57.7,-33.6) circle (3);
    \draw (57.7,-33.6) node {$p_1$};
    \draw [black] (48.5,-24.5) circle (3);
    \draw (48.5,-24.5) node {$p_0$};
    \draw [black] (57.7,-24.5) circle (3);
    \draw (57.7,-24.5) node {$p_3$};
    \draw [black] (69.8,-33.6) circle (3);
    \draw (69.8,-33.6) node {$p_2$};
    \draw [black] (79,-33.6) circle (3);
    \draw (79,-33.6) node {$p_1$};
    \draw [black] (69.8,-24.5) circle (3);
    \draw (69.8,-24.5) node {$p_0$};
    \draw [black] (79,-24.5) circle (3);
    \draw (79,-24.5) node {$p_3$};
    \draw [black] (1.4,-24.5) -- (2.6,-24.5);
    \fill [black] (2.6,-24.5) -- (1.8,-24) -- (1.8,-25);
    \draw [black] (30.1,-24.5) -- (33.3,-24.5);
    \fill [black] (33.3,-24.5) -- (32.5,-24) -- (32.5,-25);
    \draw (31.7,-24) node [above] {$a$};
    \draw [black] (36.3,-27.5) -- (36.3,-30.6);
    \fill [black] (36.3,-30.6) -- (36.8,-29.8) -- (35.8,-29.8);
    \draw (35.8,-29.05) node [left] {$a$};
    \draw [black] (33.3,-33.6) -- (30.1,-33.6);
    \fill [black] (30.1,-33.6) -- (30.9,-34.1) -- (30.9,-33.1);
    \draw (31.7,-34.1) node [below] {$a$};
    \draw [black] (51.5,-24.5) -- (54.7,-24.5);
    \fill [black] (54.7,-24.5) -- (53.9,-24) -- (53.9,-25);
    \draw (53.1,-24) node [above] {$a$};
    \draw [black] (57.7,-27.5) -- (57.7,-30.6);
    \fill [black] (57.7,-30.6) -- (58.2,-29.8) -- (57.2,-29.8);
    \draw (58.2,-29.05) node [right] {$a$};
    \draw [black] (54.7,-33.6) -- (51.5,-33.6);
    \fill [black] (51.5,-33.6) -- (52.3,-34.1) -- (52.3,-33.1);
    \draw (53.1,-34.1) node [below] {$a$};
    \draw [black] (22.4,-24.5) -- (24.1,-24.5);
    \fill [black] (24.1,-24.5) -- (23.3,-24) -- (23.3,-25);
    \draw [black] (43.7,-24.5) -- (45.5,-24.5);
    \fill [black] (45.5,-24.5) -- (44.7,-24) -- (44.7,-25);
    \draw [black] (8.6,-24.5) -- (11.9,-24.5);
    \fill [black] (11.9,-24.5) -- (11.1,-24) -- (11.1,-25);
    \draw (10.25,-24) node [above] {$a$};
    \draw [black] (14.9,-27.5) -- (14.9,-30.6);
    \fill [black] (14.9,-30.6) -- (15.4,-29.8) -- (14.4,-29.8);
    \draw (15.4,-29.05) node [right] {$a$};
    \draw [black] (11.9,-33.6) -- (8.6,-33.6);
    \fill [black] (8.6,-33.6) -- (9.4,-34.1) -- (9.4,-33.1);
    \draw (10.25,-34.1) node [below] {$a$};
    \draw [black] (5.6,-30.6) -- (5.6,-27.5);
    \fill [black] (5.6,-27.5) -- (5.1,-28.3) -- (6.1,-28.3);
    \draw (5.1,-29.05) node [left] {$a$};
    \draw [black] (29.23,-31.49) -- (34.17,-26.61);
    \fill [black] (34.17,-26.61) -- (33.25,-26.82) -- (33.95,-27.53);
    \draw (32.67,-29.53) node [below] {$a$};
    \draw [black] (50.592,-31.507) arc (118.60039:61.39961:5.239);
    \fill [black] (55.61,-31.51) -- (55.14,-30.69) -- (54.67,-31.56);
    \draw (53.1,-30.37) node [above] {$a$};
    \draw [black] (79,-27.5) -- (79,-30.6);
    \fill [black] (79,-30.6) -- (79.5,-29.8) -- (78.5,-29.8);
    \draw (79.5,-29.05) node [right] {$a$};
    \draw [black] (72.8,-24.5) -- (76,-24.5);
    \fill [black] (76,-24.5) -- (75.2,-24) -- (75.2,-25);
    \draw (74.4,-24) node [above] {$a$};
    \draw [black] (76,-33.6) -- (72.8,-33.6);
    \fill [black] (72.8,-33.6) -- (73.6,-34.1) -- (73.6,-33.1);
    \draw (74.4,-34.1) node [below] {$a$};
    \draw [black] (65.2,-24.5) -- (66.8,-24.5);
    \fill [black] (66.8,-24.5) -- (66,-24) -- (66,-25);
    \draw [black] (70.785,-30.779) arc (188.48384:-99.51616:2.25);
    \draw (75.14,-27.4) node [right] {$a$};
    \fill [black] (72.64,-32.67) -- (73.5,-33.04) -- (73.36,-32.05);

\end{tikzpicture}
\end{center}

%% file: sections/experiments.tex
~\
\section{Implementation and Evaluation}
\label{sec:eval}

\vspace{1.1em}
\subsubsection{Implementation and Experimental Setup}

We evaluate the three algorithms discussed so far:
the {\em unoptimized} algorithm \algov{\Phi}{\psi}{\gen_\textit{LTS}} of~\cite{ScenariosHVC2014,sigact2017} (Section~\ref{sec_concreteGCG});
and the {\em naive optimization} \algovpr{\Phi}{\psi}{\overset{\simeq}{\gen}}
and {\em permuter optimization} \algovpr{\Phi}{\psi}{\gen_\pi} algorithms of Section~\ref{subsec_opt}.
These are respectively labeled `unopt.', `naive opt.', and `perm. opt.' in the tables that follow.

In addition, we evaluate the unoptimized algorithm outfitted with an additional optimization, which we call the dead
transition optimization.
We say that a transition of an LTS is {\it dead} if
this transition is never taken in any run.
If $M$ with states $Q$ is correct and has $k$ dead transitions, then there are $\abs{Q}^k$ solutions that are equivalent modulo dead transitions, since we can point a dead transition anywhere while maintaining correctness. The dead transition optimization prunes all solutions which are equivalent modulo dead transitions. It is equivalent to the unoptimized algorithm in cases where there are no solutions or where we are looking for only one solution. Therefore, we evaluate the dead transition optimization side-by-side with the unoptimized solution only when we are enumerating all correct completions.
The naive and permuter optimizations both include the dead transition optimization.

We use~\cite{christos-tool}, the Python implementation of \algov{\Phi}{\psi}{\gen_\textit{LTS}}
made publicly available by the authors of~\cite{ScenariosHVC2014,sigact2017},
and we implement our optimizations on top of~\cite{christos-tool} in order to keep the comparison fair.
The tool can handle completion of distributed systems, rather than of single LTSs.
Distributed systems are represented as networks of communicating LTSs similar to those in~\cite{sigact2017}.
Specifications are represented using safety and liveness (B\"uchi) monitors, again similar to those in~\cite{sigact2017}.
However, let us again mention that our approach is not specific to any particular specification logic; it should allow for performance gains whenever the cost of model-checking is greater than the cost of the simple syntactic transformations applied by the permuter.
We use the SAT solver Z3~\cite{citeZ3} to pick candidates from the search space.
Our experimental results can be reproduced using a publicly available artifact~\cite{updated-tool}.

For our experiments we use the ABP case study as presented in~\cite{sigact2017} as well as our own two phase commit (2PC) case study.
We consider three use cases:
(1) {\em completion enumeration}: enumerate all correct completions;
(2) {\em realizable 1-completion}: return the first correct completion and stop, where we ensure that a correct completion exists;
and (3) {\em unrealizable 1-completion}: return the first correct completion, except that we ensure that none exists (and therefore the tool has to explore the entire candidate space in vain).

We consider a {\em many-process synthesis} scenario, where the goal is to synthesize two or more processes,
and a {\em 1-process synthesis} scenario, where the goal is to synthesize a single process.
In both of these scenarios across both the ABP and 2PC case studies, the synthesized processes are composed with additional environment processes and safety and liveness monitors.
The results of the many-process synthesis scenario are presented shortly.
Due to lack of space, the results of the 1-process synthesis scenario are presented in Appendix~\ref{subsec1proc}.
The
latter
results do not add much additional insight, except that 1-process synthesis tends to take less time.

Each experiment was run on a dedicated 2.40GHz CPU core located on the Northeastern Discovery Cluster.
All times are in seconds, rounded up to the second.

\subsubsection{Many-Process Synthesis Experiments}
\label{subsec2procs}

In all these experiments, there are multiple LTSs that must both be completed.
In the case of ABP:
(1) the incomplete ABP $\textit{Receiver}_0$ of Fig.~\ref{fig:sigact-fig14} without further modification;
(2) an incomplete sender process, which is obtained by removing some set of transitions from process \textit{Sender} of Fig.~\ref{fig:sndr_complete}.
The set of transitions removed from \textit{Sender} are all incoming and all outgoing transitions from all states designated as permutable for that experiment (column $A$ in the tables that follow).
For instance, in experiment $\{s_1,s_2\}$ of Table~\ref{dist-all-table}
we remove all incoming and outgoing transitions from states $s_1$ and $s_2$ of \textit{Sender}, and similarly for the other experiments.
And in the case of 2PC (see Appendix~\ref{2pc_procs} for figures):
(1) the two incomplete 2PC database managers, of Fig.~\ref{fig:incomplete-db-man} without further modification;
(2) an incomplete transaction manager, which is obtained by removing some set of transitions from process \textit{tx. man.} of Fig.~\ref{fig:complete-tx-man}.

\paragraph{Completion Enumeration}

Table~\ref{dist-all-table} presents the results for the completion enumeration use case and many-process synthesis scenario.
Columns labeled
{\em sol.} and {\em iter.} record the number of solutions (i.e., $|\Sol|$) and loop iterations of Algorithm~\ref{fig:newalgo2}
(i.e., the number of candidates $|\Cand|$, i.e., the number of times the SAT routine is called), respectively.
Pilot experiments showed negligible variance across random seeds, so reported times are for one seed.
TO denotes a timeout of 4 hours, in which case $\sfrac{p}{q}$ means the tool produced $p$
out of the total $q$ solutions.
For the dead opt. column, we know that $q = 24\cdot n$, where $n$ is the number of solutions/equivalence classes found by the permuter optimization and $24 = 4!$ is the number of isomorphisms
for 4 states.
Since the naive optimization  produces equivalence classes, $q = n$ for the naive opt. column.

\input{figs/dist-all-table}

The results in Table~\ref{dist-all-table}
are consistent with our theoretical analyses. When there are 2 permutable states, the naive and permuter optimizations
explore about half the number of candidates
as the dead transitions method. For 3 permutable states, the optimized methods explore about $3!=6$ times fewer candidates.
For 4 permutable states, the optimized methods explore about $4!=24$ times fewer candidates than the dead transitions method in the only experiment where the unoptimized method does not timeout. Notably, the permuter optimization does not timeout on any of these experiments.

\paragraph{Realizable 1-Completion}

Table~\ref{dist-one-table} presents the results for the realizable 1-completion use case (return the first solution found and stop) and many-process synthesis scenario.
Our experiments and those of~\cite{sigact2017} suggest that there is more time variability for this task depending on the random seed provided to Z3. Thus, for Table~\ref{dist-one-table} we run the tools for 10 different random seeds and report average times and number of iterations, rounded up.
In one case (last row of Table~\ref{dist-one-table}), for a single seed out of the 10 seeds,
the program timed out before finding a solution. As the true average is unknown in this case, we report it as TO.

\input{figs/dist-one-table}

\paragraph{Unrealizable 1-Completion}

Table~\ref{dist-zero-table} presents the results for the unrealizable 1-completion use case and many-process synthesis scenario.
For these experiments, we artificially modify the ABP \textit{Sender} by completely removing state $s_7$, which results in no correct completion existing. A similar change is applied to \textit{tx. man.} in the case of 2PC. Thus, the tools explore the entire search space and terminate without finding a solution.
As can be seen, the permuter optimization significantly prunes the search space and achieves considerable speedups.

\input{figs/dist-zero-table}

%% file: figs/dist-all-table.tex
\begin{table}\centering
\scalebox{0.94}{%
\begin{tabular}{|r||rrr||rrr||rrr||rrr||}
\hline
&\multicolumn{3}{c||}{unopt.}
&\multicolumn{3}{c||}{dead opt.}
&\multicolumn{3}{c||}{naive opt.}
&\multicolumn{3}{c||}{perm. opt.}\\
\hline
Case Study;$A$
& sol.          & iter.   & time
& sol.          & iter.   & time
& sol.         & iter.   & time
& sol.        & iter.   & time  \\

2PC;$\{p_1,p_2\}$ & 4 & 536 & 47 & 4 & 536 & 46 & 2 & 274 & 34 & 2 & 274 & 28 \\
2PC;$\{p_2,p_3\}$ & 48 & 1417 & 130 & 4 & 1352 & 124 & 2 & 735 & 93 & 2 & 735 & 77 \\
2PC;$\{p_3,p_4\}$ & 336 & 2852 & 266 & 6 & 2600 & 231 & 3 & 1328 & 161 & 3 & 1328 & 134 \\
2PC;$\{p_4,p_8\}$ & 576 & 1813 & 168 & 4 & 1237 & 112 & 2 & 575 & 75 & 2 & 648 & 66 \\
ABP;$\{s_1,s_2\}$ & 64 & 628 & 27 & 8 & 574 & 21 & 4 & 289 & 18 & 4 & 304 & 12 \\
ABP;$\{s_2,s_3\}$ & 64 & 1859 & 75 & 8 & 1832 & 70 & 4 & 946 & 55 & 4 & 943 & 37 \\
ABP;$\{s_3,s_4\}$ & 32 & 374 & 18 & 4 & 353 & 13 & 2 & 188 & 12 & 2 & 192 & 8 \\
ABP;$\{s_4,s_5\}$ & 32 & 3728 & 177 & 4 & 3638 & 170 & 2 & 1913 & 160 & 2 & 1833 & 93 \\
ABP;$\{s_5,s_6\}$ & 64 & 449 & 27 & 8 & 412 & 21 & 4 & 199 & 18 & 4 & 201 & 11 \\
ABP;$\{s_6,s_7\}$ & 64 & 1518 & 94 & 8 & 1481 & 87 & 4 & 769 & 80 & 4 & 752 & 47 \\
2PC;$\{p_2,p_3,p_4\}$ & 2016 & 17478 & 1896 & 36 & 15646 & 1677 & 6 & 2693 & 719 & 6 & 2693 & 466 \\
2PC;$\{p_3,p_4,p_8\}$ & \sfrac{79391}{-} & 101278 & TO & 36 & 23044 & 2498 & 6 & 4079 & 1064 & 6 & 3997 & 682 \\
ABP;$\{s_1,s_2,s_3\}$ & 192 & 5641 & 226 & 24 & 5499 & 207 & 4 & 968 & 155 & 4 & 937 & 49 \\
ABP;$\{s_2,s_3,s_4\}$ & 3072 & 23025 & 1470 & 48 & 19114 & 934 & 8 & 3639 & 722 & 8 & 3331 & 225 \\
ABP;$\{s_3,s_4,s_5\}$ & 96 & 14651 & 748 & 12 & 15108 & 760 & 2 & 2599 & 567 & 2 & 2520 & 172 \\
ABP;$\{s_4,s_5,s_6\}$ & 1536 & 14405 & 876 & 24 & 13269 & 686 & 4 & 2458 & 554 & 4 & 2215 & 151 \\
ABP;$\{s_5,s_6,s_7\}$ & 192 & 4686 & 287 & 24 & 4559 & 268 & 4 & 809 & 241 & 4 & 748 & 57 \\
2PC;$\{p_1,p_2,p_3,p_4\}$ & \sfrac{8064}{-} & 70250 & TO & 144 & 62280 & 11915 & 6 & 2770 & 2844 & 6 & 2719 & 1564 \\
ABP;$\{s_1,s_2,s_3,s_4\}$ & 12288 & 90031 & 8143 & 192 & 76591 & 5458 & 8 & 3704 & 2931 & 8 & 3271 & 628 \\
ABP;$\{s_3,s_4,s_5,s_6\}$ & 6144 & 59838 & 4777 & 96 & 52935 & 3543 & 4 & 2896 & 2655 & 4 & 2351 & 431 \\
ABP;$\{s_4,s_5,s_6,s_7\}$ & \sfrac{1009}{-} & 108929 & TO & \sfrac{38}{96} & 111834 & TO & \sfrac{2}{4} & 10443 & TO & 4 & 8639 & 7480 \\

\hline
\end{tabular}
}%
\caption{Many-Process Synthesis, Completion Enumeration
}
\label{dist-all-table}
\end{table}

%% file: figs/dist-one-table.tex
\begin{table}\centering
\begin{tabular}{|r||rr||rr||rr||}
\hline
&\multicolumn{2}{c||}{unopt.}&\multicolumn{2}{c||}{naive opt.}&\multicolumn{2}{c||}{perm. opt.}\\
\hline
Case Study;$A$                     & iter.          & time  & iter.          & time  & iter.         & time \\

 2PC;$\{p_1,p_2\}$ & 199 & 19 & 157 & 20 & 157 & 17 \\
2PC;$\{p_2,p_3\}$ & 483 & 47 & 429 & 55 & 426 & 46 \\
2PC;$\{p_3,p_4\}$ & 798 & 72 & 696 & 84 & 666 & 69 \\
2PC;$\{p_4,p_8\}$ & 380 & 37 & 319 & 44 & 311 & 34 \\
ABP;$\{s_1,s_2\}$ & 111 & 4 & 110 & 7 & 100 & 4 \\
ABP;$\{s_2,s_3\}$ & 220 & 9 & 205 & 13 & 200 & 9 \\
ABP;$\{s_3,s_4\}$ & 106 & 5 & 102 & 7 & 105 & 5 \\
ABP;$\{s_4,s_5\}$ & 1669 & 75 & 909 & 73 & 1202 & 60 \\
ABP;$\{s_5,s_6\}$ & 102 & 5 & 95 & 8 & 102 & 5 \\
ABP;$\{s_6,s_7\}$ & 507 & 28 & 294 & 28 & 294 & 17 \\
2PC;$\{p_2,p_3,p_4\}$ & 440 & 48 & 590 & 147 & 561 & 89 \\
2PC;$\{p_3,p_4,p_8\}$ & 954 & 94 & 861 & 205 & 796 & 121 \\
ABP;$\{s_1,s_2,s_3\}$ & 332 & 12 & 225 & 36 & 240 & 13 \\
ABP;$\{s_2,s_3,s_4\}$ & 2462 & 108 & 904 & 170 & 1028 & 64 \\
ABP;$\{s_3,s_4,s_5\}$ & 2267 & 102 & 1040 & 219 & 819 & 52 \\
ABP;$\{s_4,s_5,s_6\}$ & 2735 & 130 & 1513 & 333 & 1327 & 92 \\
ABP;$\{s_5,s_6,s_7\}$ & 361 & 21 & 264 & 69 & 308 & 22 \\

2PC;$\{p_1,p_2,p_3,p_4\}$ & 806 & 81 & 495 & 387 & 572 & 220 \\
ABP;$\{s_1,s_2,s_3,s_4\}$ & 1957 & 85 & 1068 & 760 & 890 & 122 \\
ABP;$\{s_3,s_4,s_5,s_6\}$ & 5425 & 261 & 1003 & 860 & 1601 & 234 \\
ABP;$\{s_4,s_5,s_6,s_7\}$ & 16098 & 1088 & TO & TO & 4159 & 1158 \\

\hline
\end{tabular}
\caption{Many-Process Synthesis, Realizable 1-Completion}
\label{dist-one-table}
\vspace{-3em}
\end{table}

%% file: figs/dist-zero-table.tex
\begin{table}\centering
\begin{tabular}{|r||rr||rr||rr||}
\hline
&\multicolumn{2}{c||}{unopt.}&\multicolumn{2}{c||}{naive opt.}&\multicolumn{2}{c||}{perm. opt.}\\
\hline
Case Study;$A$                     & iter.          & time  & iter.          & time  & iter.         & time  \\

 2PC;$\{p_1,p_2\}$ & 3207 & 292 & 1658 & 206 & 1655 & 175 \\
 2PC;$\{p_2,p_3\}$ & 9792 & 978 & 4996 & 646 & 4982 & 552 \\
 2PC;$\{p_3,p_4\}$ & 14911 & 1527 & 7645 & 1053 & 7589 & 878 \\
 2PC;$\{p_4,p_8\}$ & 5123 & 494 & 2537 & 339 & 2555 & 282 \\
 ABP;$\{s_1,s_2\}$ & 1650 & 58 & 879 & 52 & 853 & 33 \\
 ABP;$\{s_2,s_3\}$ & 4300 & 173 & 2384 & 171 & 2374 & 106 \\
 ABP;$\{s_3,s_4\}$ & 327 & 13 & 173 & 11 & 164 & 7 \\
 ABP;$\{s_4,s_5\}$ & 3108 & 143 & 1592 & 130 & 1710 & 89 \\
 ABP;$\{s_5,s_6\}$ & 333 & 16 & 172 & 15 & 168 & 9 \\
 2PC;$\{p_2,p_3,p_4\}$ & 66088 & TO & 19717 & 10867 & 19850 & 9610 \\
 2PC;$\{p_3,p_4,p_8\}$ & 70586 & TO & 26343 & TO & 26516 & 14340 \\
 ABP;$\{s_1,s_2,s_3\}$ & 20858 & 1022 & 3705 & 798 & 3668 & 253 \\
 ABP;$\{s_2,s_3,s_4\}$ & 58974 & 4021 & 10516 & 2673 & 10496 & 1052 \\
 ABP;$\{s_3,s_4,s_5\}$ & 12323 & 596 & 2231 & 504 & 2167 & 146 \\
 ABP;$\{s_4,s_5,s_6\}$ & 11210 & 557 & 2104 & 491 & 1985 & 136 \\

2PC;$\{p_1,p_2,p_3,p_4\}$ & 67659 & TO & 10365 & TO & 12308 & TO \\
ABP;$\{s_1,s_2,s_3,s_4\}$ & 129264 & TO & 12096 & TO & 14739 & TO \\
ABP;$\{s_3,s_4,s_5,s_6\}$ & 45056 & 2869 & 2466 & 2392 & 2004 & 339 \\

\hline
\end{tabular}
\caption{Many-Process Synthesis, Unrealizable 1-Completion
}
\label{dist-zero-table}
\end{table}

%% file: sections/related.tex
\section{Related Work}
\label{sec:related}

\paragraph{Synthesis of Distributed Protocols:}

Distributed system synthesis has been studied both in the reactive synthesis setting~\cite{PnueliRosner90} and in the setting of discrete-event
systems~\cite{Thistle2005,TripakisIPL}.
More recently, synthesis of distributed protocols has been studied using completion techniques
in~\cite{ScenariosHVC2014,CompletionCAV2015,sigact2017,DBLP:journals/corr/abs-2208-12400}.
\cite{ScenariosHVC2014,sigact2017} study completion of finite-state protocols such as ABP but they do not focus on enumeration.
\cite{CompletionCAV2015} considers infinite-state protocols and focus on synthesis of symbolic expressions (guards and assignments).
None of~\cite{ScenariosHVC2014,CompletionCAV2015,sigact2017} propose any reduction techniques. We propose reduction modulo isomorphisms.

\cite{DBLP:journals/corr/abs-2208-12400} studies synthesis for a class of parameterized distributed agreement-based protocols for which verification is efficiently decidable.
Another version of the paper~\cite{Jaber} considers permutations of process indices.
These are different from our permutations over process states.

Synthesis of parameterized distributed systems is also studied in~\cite{DBLP:journals/acta/MirzaieFJB20} using the notion of {\em cutoffs}, which guarantee that if a property holds for all systems up to a certain size (the cutoff size) then it also holds for systems of any size.
Cutoffs are different from our isomorphism reductions.

\paragraph{Bounded Synthesis:}

The bounded synthesis approach~\cite{FinkbeinerScheweSTTT13} limits the search space of synthesis by setting an upper bound on certain system parameters, and encodes the resulting problem into a satisfiability problem. Bounded synthesis is applicable to many application domains, including distributed system synthesis, and has been successfully used to synthesize systems such as distributed arbiters and dining philosophers~\cite{FinkbeinerScheweSTTT13}. Symmetries have also been exploited in bounded synthesis. Typically, such symmetries encode similarity of processes (e.g., all processes having the same state-transition structure, as in the case of dining philosophers). As such, these symmetries are similar to those exploited in parameterized systems, and different from our LTS isomorphisms.

\paragraph{Symmetry Reductions in Model-Checking:}

Symmetries have been exploited in model-checking~\cite{DBLP:conf/cav/ClarkeEJS98}. The basic idea is to take a model $M$ and construct a new model $M_G$ which has a much smaller state space. This construction exploits the fact that many states in $M$ might be functionally equivalent, in the sense
of
incoming and outgoing transitions.
The key distinction between this work and ours is that our symmetries are over the space of models rather than the space of states of a fixed model.
This distinction allows us to exploit symmetries for completion enumeration rather than model-checking.

\paragraph{Symmetry-Breaking Predicates:}

Symmetry-breaking
predicates
have been used to solve
SAT~\cite{DBLP:conf/kr/CrawfordGLR96},
SMT~\cite{DBLP:journals/corr/abs-1908-00860},
and even graph search problems~\cite{DBLP:conf/synasc/Heule16},
more efficiently.
Our work is related in the sense that we are also trying to prune a search space.
But our approach differs both in the notion of symmetry used (LTS isomorphism) as well as the application domain (distributed protocols).
Moreover,
rather than trying to eliminate all but one member of each equivalence class at the outset, say, by somehow adding a global
(and often prohibitively large) symmetry-breaking formula $\Xi$ to $\Phi$, we do so {\em on-the-fly} for each candidate solution.

\paragraph{Canonical Forms:}

In program synthesis
work~\cite{DBLP:conf/vmcai/SmithA19}, a candidate program is only checked for correctness if it is in some normal form.
\cite{DBLP:conf/vmcai/SmithA19} is not about synthesis of distributed protocols, and as such the normal forms considered there are very different from our LTS isomorphisms.
In particular, as with symmetry-breaking, the normal forms used in~\cite{DBLP:conf/vmcai/SmithA19} are global, defined {\em a-priori} for the entire program domain, whereas our generalizations are computed on-the-fly.
Moreover, the approach of~\cite{DBLP:conf/vmcai/SmithA19} may still
generate two equivalent programs as candidates (prior to verification), i.e., it does not satisfy our Property~\ref{inv:cand}.

\paragraph{Sketching, CEGIS, OGIS, Sciduction:}

Completion algorithms such as GCG belong to the same family of techniques as
sketching~\cite{SolarLezamaSTTT13},
counter-example guided inductive synthesis (CEGIS)~\cite{Sygus13,GulwaniPolozovSingh2017,SolarLezama2006,SolarLezamaSTTT13},
oracle-guided inductive synthesis (OGIS)~\cite{DBLP:journals/acta/JhaS17},
and sciduction~\cite{DBLP:conf/dac/Seshia12}.

%% file: sections/end.tex
\section{Conclusions}
\label{sec:end}

We proposed a novel distributed protocol synthesis approach based on completion enumeration modulo isomorphisms.
Our approach follows the {\em guess-check-generalize} synthesis paradigm, and relies on non-trivial optimizations of the generalization step that exploit state permutations.
These optimizations allow to significantly prune the search space of candidate completions,
achieving speedups of factors approximately 2 to 10 and in some cases completing experiments in minutes instead of hours.
To our knowledge, ours is the only work on distributed protocol enumeration using reductions such as isomorphism.

As future work, we plan to employ this optimized enumeration approach for the synthesis of distributed protocols that achieve not only correctness, but also performance objectives.
We also plan to address the question {\em where do the incomplete processes come from?}
If not provided by the user, such incomplete processes might be automatically generated from example scenarios as in~\cite{ScenariosHVC2014,sigact2017}, or might simply be ``empty skeletons'' of states, without any transitions.
We also plan to extend our approach to infinite-state protocols, as well as application domains beyond protocols,
as Algorithm~\ref{fig:newalgo2} is generic and thus
applicable to a wide class of synthesis domains.

%% file: sections/appendix.tex

\appendix
\section{Appendix}

\subsection{The 2PC Processes}
\label{2pc_procs}

We used the 2PC transaction manager of Fig.~\ref{fig:complete-tx-man} and the incomplete database manager of Fig.~\ref{fig:incomplete-db-man} to obtain our experimental results (cf. page~\pageref{subsec2procs}).

\begin{figure}[H]
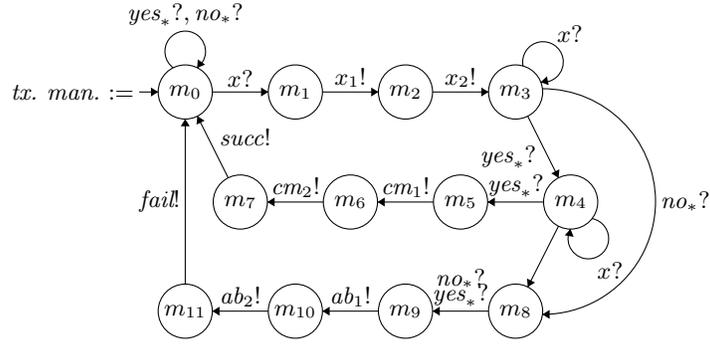

    \include{figs/2pc_H}
    \caption{A completed transaction manager}
    \label{fig:complete-tx-man}
\end{figure}


\begin{figure}[H]
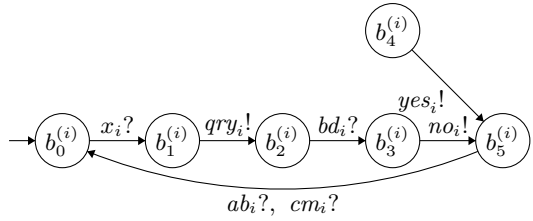

    \include{figs/2pc_B0}
    \caption{An incomplete database manager; $i\in\{1,2\}$
    }
    \label{fig:incomplete-db-man}
\end{figure}

\subsection{1-Process Synthesis Experiments}
\label{subsec1proc}

In each of these experiments, we only complete an incomplete sender process, obtained in the same manner as for the many-process experiments on page~\pageref{subsec2procs}, namely, by removing some set of transitions from process \textit{Sender} of Fig.~\ref{fig:sndr_complete}.
The set of transitions removed depends on which states are designated as permutable for that experiment.
For instance, for $A = \{s_1,s_2\}$, we remove all incoming and outgoing transitions from both states $s_1$ and $s_2$,
and similarly for the other experiments. We use the completed ABP \textit{Receiver} shown in Fig.~\ref{fig:complete-abp-rec}.

The results for the three use cases (completion enumeration, realizable 1-completion, and unrealizable 1-completion) are presented in Tables~\ref{non-all-table},~\ref{non-one-table}, and~\ref{non-zero-table}, respectively.
Similar speedups can be observed as those in the many-process synthesis scenario (page~\pageref{subsec2procs}).

\begin{figure}[H]
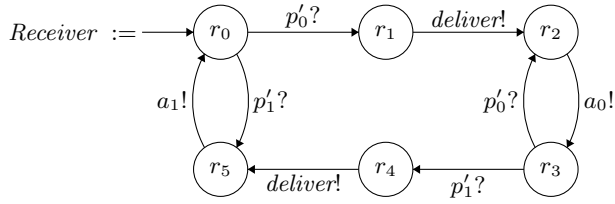

    \include{figs/complete-abp-rec}
    \caption{A completed ABP Receiver}
    \label{fig:complete-abp-rec}
\end{figure}

\input{figs/non-all-table}

~\input{figs/non-one-table}

~\input{figs/non-zero-table}

%% file: figs/2pc_H.tex
\begin{center}
\begin{tikzpicture}[scale=\tiksize]
\tikzstyle{every node}+=[inner sep=0pt]
\draw [black] (9.2,-14.6) circle (3);
\draw (9.2,-14.6) node {$m_0$};
\draw [black] (21.4,-14.6) circle (3);
\draw (21.4,-14.6) node {$m_1$};
\draw [black] (33.6,-14.6) circle (3);
\draw (33.6,-14.6) node {$m_2$};
\draw [black] (45.8,-14.6) circle (3);
\draw (45.8,-14.6) node {$m_3$};
\draw [black] (39.7,-26.8) circle (3);
\draw (39.7,-26.8) node {$m_5$};
\draw [black] (51.9,-26.8) circle (3);
\draw (51.9,-26.8) node {$m_4$};
\draw [black] (27.5,-26.8) circle (3);
\draw (27.5,-26.8) node {$m_6$};
\draw [black] (15.3,-26.8) circle (3);
\draw (15.3,-26.8) node {$m_7$};
\draw [black] (9.2,-38.9) circle (3);
\draw (9.2,-38.9) node {$m_{11}$};
\draw [black] (21.4,-38.9) circle (3);
\draw (21.4,-38.9) node {$m_{10}$};
\draw [black] (33.6,-38.9) circle (3);
\draw (33.6,-38.9) node {$m_9$};
\draw [black] (45.8,-38.9) circle (3);
\draw (45.8,-38.9) node {$m_8$};
\draw [black] (4.1,-14.6) -- (6.2,-14.6);
\draw (3.6,-14.6) node [left] {\textit{tx. man.}\ :=};
\fill [black] (6.2,-14.6) -- (5.4,-14.1) -- (5.4,-15.1);
\draw [black] (12.2,-14.6) -- (18.4,-14.6);
\fill [black] (18.4,-14.6) -- (17.6,-14.1) -- (17.6,-15.1);
\draw (15.3,-14.1) node [above] {$x?$};
\draw [black] (24.4,-14.6) -- (30.6,-14.6);
\fill [black] (30.6,-14.6) -- (29.8,-14.1) -- (29.8,-15.1);
\draw (27.5,-14.1) node [above] {$x_1!$};
\draw [black] (36.6,-14.6) -- (42.8,-14.6);
\fill [black] (42.8,-14.6) -- (42,-14.1) -- (42,-15.1);
\draw (39.7,-14.1) node [above] {$x_2!$};
\draw [black] (48.771,-14.238) arc (90.0819:-90.0819:12.512);
\fill [black] (48.77,-39.26) -- (49.57,-39.76) -- (49.57,-38.76);
\draw (61.8,-26.75) node [right] {$\textit{no}_*?$};
\draw [black] (47.14,-17.28) -- (50.56,-24.12);
\fill [black] (50.56,-24.12) -- (50.65,-23.18) -- (49.75,-23.62);
\draw (48.15,-21.81) node [left] {$\textit{yes}_*?$};
\draw [black] (50.55,-29.48) -- (47.15,-36.22);
\fill [black] (47.15,-36.22) -- (47.96,-35.73) -- (47.06,-35.28);
\draw [black] (48.9,-26.8) -- (42.7,-26.8);
\fill [black] (42.7,-26.8) -- (43.5,-27.3) -- (43.5,-26.3);
\draw (45.8,-26.3) node [above] {$\textit{yes}_*?$};
\draw [black] (36.7,-26.8) -- (30.5,-26.8);
\fill [black] (30.5,-26.8) -- (31.3,-27.3) -- (31.3,-26.3);
\draw (33.6,-26.3) node [above] {$\textit{cm}_1!$};
\draw [black] (24.5,-26.8) -- (18.3,-26.8);
\fill [black] (18.3,-26.8) -- (19.1,-27.3) -- (19.1,-26.3);
\draw (21.4,-26.3) node [above] {$\textit{cm}_2!$};
\draw [black] (13.96,-24.12) -- (10.54,-17.28);
\fill [black] (10.54,-17.28) -- (10.45,-18.22) -- (11.35,-17.78);
\draw (12.95,-19.59) node [right] {$\textit{succ}!$};
\draw [black] (9.2,-35.9) -- (9.2,-17.6);
\fill [black] (9.2,-17.6) -- (8.7,-18.4) -- (9.7,-18.4);
\draw (8.7,-26.75) node [left] {$\textit{fail}!$};
\draw [black] (18.4,-38.9) -- (12.2,-38.9);
\fill [black] (12.2,-38.9) -- (13,-39.4) -- (13,-38.4);
\draw (15.3,-38.4) node [above] {$\textit{ab}_2!$};
\draw [black] (30.6,-38.9) -- (24.4,-38.9);
\fill [black] (24.4,-38.9) -- (25.2,-39.4) -- (25.2,-38.4);
\draw (27.5,-38.4) node [above] {$\textit{ab}_1!$};
\draw [black] (42.8,-38.9) -- (36.6,-38.9);
\fill [black] (36.6,-38.9) -- (37.4,-39.4) -- (37.4,-38.4);
\draw (39.7,-38.4) node [above] {$\textit{yes}_*?$};
\draw (39.7,-36.4) node [above] {$\textit{no}_*?$};
\draw [black] (46.629,-11.729) arc (191.63248:-96.36752:2.25);
\draw (51.81,-9.13) node [above] {$x?$};
\fill [black] (48.58,-13.51) -- (49.47,-13.84) -- (49.27,-12.86);
\draw [black] (7.877,-11.92) arc (234:-54:2.25);
\draw (9.2,-7.35) node [above] {$\textit{yes}_*?,\textit{no}_*?$};
\fill [black] (10.52,-11.92) -- (11.4,-11.57) -- (10.59,-10.98);
\draw [black] (54.257,-28.637) arc (79.80895:-208.19105:2.25);
\draw (56.27,-33.39) node [below] {$x?$};
\fill [black] (51.88,-29.79) -- (51.24,-30.49) -- (52.23,-30.66);
\end{tikzpicture}
\end{center}

%% file: figs/2pc_B0.tex
\begin{center}
\begin{tikzpicture}[scale=\tiksize]
\tikzstyle{every node}+=[inner sep=0pt]
\draw [black] (13,-27.6) circle (3);
\draw (13,-27.6) node {$b_0^{(i)}$};
\draw [black] (25.2,-27.6) circle (3);
\draw (25.2,-27.6) node {$b_1^{(i)}$};
\draw [black] (37.4,-27.6) circle (3);
\draw (37.4,-27.6) node {$b_2^{(i)}$};
\draw [black] (61.8,-27.6) circle (3);
\draw (61.8,-27.6) node {$b_5^{(i)}$};
\draw [black] (49.6,-27.6) circle (3);
\draw (49.6,-27.6) node {$b_3^{(i)}$};
\draw [black] (49.6,-15.4) circle (3);
\draw (49.6,-15.4) node {$b_4^{(i)}$};
\draw [black] (16,-27.6) -- (22.2,-27.6);
\fill [black] (22.2,-27.6) -- (21.4,-27.1) -- (21.4,-28.1);
\draw (19.1,-27.1) node [above] {$x_i?$};
\draw [black] (28.2,-27.6) -- (34.4,-27.6);
\fill [black] (34.4,-27.6) -- (33.6,-27.1) -- (33.6,-28.1);
\draw (31.3,-27.1) node [above] {$\textit{qry}_i!$};
\draw [black] (40.4,-27.6) -- (46.6,-27.6);
\fill [black] (46.6,-27.6) -- (45.8,-27.1) -- (45.8,-28.1);
\draw (43.5,-27.1) node [above] {$\textit{bd}_i?$};
\draw [black] (52.6,-27.6) -- (58.8,-27.6);
\fill [black] (58.8,-27.6) -- (58,-27.1) -- (58,-28.1);
\draw (55.7,-27.1) node [above] {$\textit{no}_i!$};
\draw [black] (51.72,-17.52) -- (59.68,-25.48);
\fill [black] (59.68,-25.48) -- (59.47,-24.56) -- (58.76,-25.27);
\draw (52.73,-21.98) node [below] {$\textit{yes}_i!$};
\draw [black] (59.046,-28.789) arc (-68.12406:-111.87594:58.095);
\fill [black] (15.75,-28.79) -- (16.31,-29.55) -- (16.68,-28.62);
\draw (37.4,-33.47) node [below] {$\textit{ab}_i?,\mbox{ }\textit{cm}_i?$};
\draw [black] (7.1,-27.6) -- (10,-27.6);
\fill [black] (10,-27.6) -- (9.2,-27.1) -- (9.2,-28.1);
\end{tikzpicture}
\end{center}

%% file: figs/complete-abp-rec.tex

\begin{center}
\begin{tikzpicture}[scale=\tiksize]
\tikzstyle{every node}+=[inner sep=0pt]
\draw [black] (22.5,-33.2) circle (3);
\draw (22.5,-33.2) node {$r_0$};
\draw [black] (40.8,-33.2) circle (3);
\draw (40.8,-33.2) node {$r_1$};
\draw [black] (59.1,-33.2) circle (3);
\draw (59.1,-33.2) node {$r_2$};
\draw [black] (22.5,-48.4) circle (3);
\draw (22.5,-48.4) node {$r_5$};
\draw [black] (40.8,-48.4) circle (3);
\draw (40.8,-48.4) node {$r_4$};
\draw [black] (59.1,-48.4) circle (3);
\draw (59.1,-48.4) node {$r_3$};
\draw [black] (25.5,-33.2) -- (37.8,-33.2);
\fill [black] (37.8,-33.2) -- (37,-32.7) -- (37,-33.7);
\draw (31.65,-32.7) node [above] {$p'_0?$};
\draw [black] (43.8,-33.2) -- (56.1,-33.2);
\fill [black] (56.1,-33.2) -- (55.3,-32.7) -- (55.3,-33.7);
\draw (49.95,-32.7) node [above] {$\textit{deliver}!$};
\draw [black] (60.871,-35.61) arc (28.43666:-28.43666:10.9);
\fill [black] (60.87,-45.99) -- (61.69,-45.53) -- (60.81,-45.05);
\draw (62.69,-40.8) node [right] {$a_0!$};
\draw [black] (20.721,-45.996) arc (-151.41152:-208.58848:10.859);
\fill [black] (20.72,-35.6) -- (19.9,-36.07) -- (20.78,-36.55);
\draw (18.9,-40.8) node [left] {$a_1!$};
\draw [black] (56.1,-48.4) -- (43.8,-48.4);
\fill [black] (43.8,-48.4) -- (44.6,-48.9) -- (44.6,-47.9);
\draw (49.95,-48.9) node [below] {$p'_1?$};
\draw [black] (24.271,-35.61) arc (28.43666:-28.43666:10.9);
\fill [black] (24.27,-45.99) -- (25.09,-45.53) -- (24.21,-45.05);
\draw (26.09,-40.8) node [right] {$p'_1?$};
\draw [black] (57.353,-45.973) arc (-152.04582:-207.95418:11.034);
\fill [black] (57.35,-35.63) -- (56.54,-36.1) -- (57.42,-36.57);
\draw (55.57,-40.8) node [left] {$p'_0?$};
\draw [black] (13.8,-33.2) -- (19.5,-33.2);
\draw (13.3,-33.2) node [left] {$\textit{Receiver}\mbox{ }:=$};
\fill [black] (19.5,-33.2) -- (18.7,-32.7) -- (18.7,-33.7);
\draw [black] (37.8,-48.4) -- (25.5,-48.4);
\fill [black] (25.5,-48.4) -- (26.3,-48.9) -- (26.3,-47.9);
\draw (31.65,-48.9) node [below] {$\textit{deliver}!$};
\end{tikzpicture}
\end{center}

%% file: figs/non-all-table.tex
\begin{table}\centering
\scalebox{0.90}{%
\begin{tabular}{|r||rrr||rrr||rrr||rrr||}
\hline
&\multicolumn{3}{c||}{unopt.}&\multicolumn{3}{c||}{dead opt.}&\multicolumn{3}{c||}{naive opt.}&\multicolumn{3}{c||}{perm. opt.}\\
\hline
Case Study;
$A$
& sol.         & iter.   & time
& sol.         & iter.   & time
& sol.         & iter.   & time
& sol.        & iter.  & time  \\

 2PC;$\{p_1,p_2\}$ & 4 & 493 & 49 & 4 & 493 & 44 & 2 & 252 & 32 & 2 & 252 & 27 \\
 2PC;$\{p_2,p_3\}$ & 48 & 1121 & 113 & 4 & 1092 & 103 & 2 & 571 & 76 & 2 & 580 & 65 \\
 2PC;$\{p_3,p_4\}$ & 336 & 2247 & 208 & 6 & 1751 & 160 & 3 & 908 & 117 & 3 & 921 & 105 \\
 2PC;$\{p_4,p_8\}$ & 576 & 1587 & 150 & 4 & 974 & 91 & 2 & 492 & 67 & 2 & 492 & 52 \\
 ABP;$\{s_1,s_2\}$ & 64 & 391 & 20 & 8 & 336 & 13 & 4 & 171 & 12 & 4 & 171 & 7 \\
 ABP;$\{s_2,s_3\}$ & 64 & 1329 & 57 & 8 & 1277 & 51 & 4 & 673 & 45 & 4 & 646 & 27 \\
 ABP;$\{s_3,s_4\}$ & 32 & 363 & 16 & 4 & 306 & 11 & 2 & 160 & 10 & 2 & 166 & 7 \\
 ABP;$\{s_4,s_5\}$ & 32 & 3729 & 175 & 4 & 3686 & 169 & 2 & 1825 & 149 & 2 & 1786 & 88 \\
 ABP;$\{s_5,s_6\}$ & 64 & 383 & 25 & 8 & 333 & 17 & 4 & 166 & 16 & 4 & 167 & 10 \\
 ABP;$\{s_6,s_7\}$ & 64 & 1324 & 84 & 8 & 1245 & 76 & 4 & 634 & 71 & 4 & 628 & 40 \\
 2PC;$\{p_2,p_3,p_4\}$ & 2016 & 12793 & 1363 & 36 & 10650 & 1139 & 6 & 1842 & 509 & 6 & 1907 & 355 \\
 2PC;$\{p_3,p_4,p_8\}$ & \sfrac{87774}{-} & 103999 & TO & 36 & 14747 & 1551 & 6 & 2491 & 667 & 6 & 2484 & 439 \\
 ABP;$\{s_1,s_2,s_3\}$ & 192 & 3986 & 171 & 24 & 3845 & 153 & 4 & 673 & 116 & 4 & 652 & 35 \\
 ABP;$\{s_2,s_3,s_4\}$ & 3072 & 18504 & 1197 & 48 & 15790 & 775 & 8 & 2959 & 631 & 8 & 2665 & 196 \\
 ABP;$\{s_3,s_4,s_5\}$ & 96 & 14446 & 700 & 12 & 14083 & 675 & 2 & 2494 & 545 & 2 & 2508 & 165 \\
 ABP;$\{s_4,s_5,s_6\}$ & 1536 & 14422 & 866 & 24 & 12477 & 609 & 4 & 2382 & 543 & 4 & 2301 & 161 \\
 ABP;$\{s_5,s_6,s_7\}$ & 192 & 3970 & 257 & 24 & 3731 & 228 & 4 & 657 & 203 & 4 & 636 & 48 \\
 2PC;$\{p_1,p_2,p_3,p_4\}$ & 8064 & 50868 & 9354 & 144 & 42926 & 7557 & 6 & 1905 & 1956 & 6 & 1934 & 1124 \\
ABP;$\{s_1,s_2,s_3,s_4\}$ & 12288 & 74398 & 6270 & 192 & 61456 & 3975 & 8 & 3195 & 2743 & 8 & 2641 & 449 \\
ABP;$\{s_3,s_4,s_5,s_6\}$ & 6144 & 58044 & 4400 & 96 & 50824 & 3276 & 4 & 2554 & 2422 & 4 & 2130 & 351 \\
ABP;$\{s_4,s_5,s_6,s_7\}$ & \sfrac{512}{-} & 112328 & TO & \sfrac{24}{96} & 112762 & TO & \sfrac{4}{4} & 9963 & TO & 4 & 8144 & 6181 \\

\hline
\end{tabular}
}
\caption{1-Process Synthesis, Completion Enumeration
}
\label{non-all-table}
\end{table}

%% file: figs/non-one-table.tex
\begin{table}\centering
\begin{tabular}{|r||rr||rr||rr||}
\hline
&\multicolumn{2}{c||}{unopt.}&\multicolumn{2}{c||}{naive opt.}&\multicolumn{2}{c||}{perm. opt.}\\
\hline
Case Study;$A$                     & iter.          & time & iter.          & time  & iter.         & time \\

 2PC;$\{p_1,p_2\}$ & 185 & 18 & 125 & 17 & 125 & 14 \\
2PC;$\{p_2,p_3\}$ & 326 & 34 & 307 & 41 & 274 & 32 \\
2PC;$\{p_3,p_4\}$ & 387 & 39 & 408 & 52 & 395 & 43 \\
2PC;$\{p_4,p_8\}$ & 200 & 22 & 234 & 34 & 228 & 26 \\
ABP;$\{s_1,s_2\}$ & 91 & 3 & 83 & 5 & 79 & 3 \\
ABP;$\{s_2,s_3\}$ & 314 & 12 & 244 & 16 & 193 & 9 \\
ABP;$\{s_3,s_4\}$ & 99 & 4 & 76 & 5 & 83 & 4 \\
ABP;$\{s_4,s_5\}$ & 878 & 40 & 1395 & 111 & 939 & 46 \\
ABP;$\{s_5,s_6\}$ & 61 & 4 & 59 & 5 & 67 & 4 \\
ABP;$\{s_6,s_7\}$ & 395 & 22 & 288 & 30 & 301 & 18 \\
2PC;$\{p_2,p_3,p_4\}$ & 366 & 41 & 410 & 105 & 336 & 60 \\
2PC;$\{p_3,p_4,p_8\}$ & 747 & 76 & 849 & 211 & 557 & 90 \\
ABP;$\{s_1,s_2,s_3\}$ & 384 & 14 & 287 & 46 & 261 & 14 \\
ABP;$\{s_2,s_3,s_4\}$ & 2114 & 96 & 744 & 145 & 583 & 37 \\
ABP;$\{s_3,s_4,s_5\}$ & 1273 & 58 & 945 & 202 & 1320 & 83 \\
ABP;$\{s_4,s_5,s_6\}$ & 2466 & 114 & 1356 & 304 & 1095 & 72 \\
ABP;$\{s_5,s_6,s_7\}$ & 249 & 16 & 188 & 53 & 191 & 15 \\

 2PC;$\{p_1,p_2,p_3,p_4\}$ & 366 & 41 & 403 & 317 & 350 & 138 \\
ABP;$\{s_1,s_2,s_3,s_4\}$ & 2555 & 117 & 803 & 615 & 723 & 95 \\
ABP;$\{s_3,s_4,s_5,s_6\}$ & 3147 & 148 & 1525 & 1317 & 1161 & 169 \\
ABP;$\{s_4,s_5,s_6,s_7\}$ & 23353 & 1660 & 4082 & 4916 & 3969 & 904 \\

\hline
\end{tabular}
\caption{1-Process Synthesis, Realizable 1-Completion}
\label{non-one-table}
\end{table}

%% file: figs/non-zero-table.tex
\begin{table}\centering
\begin{tabular}{|r||rr||rr||rr||}
\hline
&\multicolumn{2}{c||}{unopt.}&\multicolumn{2}{c||}{naive opt.}&\multicolumn{2}{c||}{perm. opt.}\\
\hline
Case Study;$A$                     & iter.          & time   & iter.          & time  & iter.         & time  \\

 2PC;$\{p_1,p_2\}$ & 3184 & 294 & 1638 & 211 & 1625 & 172 \\
2PC;$\{p_2,p_3\}$ & 8916 & 926 & 4645 & 632 & 4598 & 535 \\
2PC;$\{p_3,p_4\}$ & 13659 & 1410 & 6729 & 917 & 6820 & 814 \\
2PC;$\{p_4,p_8\}$ & 4609 & 456 & 2320 & 320 & 2317 & 257 \\
ABP;$\{s_1,s_2\}$ & 1229 & 47 & 630 & 43 & 628 & 26 \\
ABP;$\{s_2,s_3\}$ & 3405 & 149 & 1788 & 136 & 1728 & 81 \\
ABP;$\{s_3,s_4\}$ & 290 & 11 & 150 & 10 & 156 & 7 \\
ABP;$\{s_4,s_5\}$ & 3171 & 145 & 1565 & 130 & 1610 & 80 \\
ABP;$\{s_5,s_6\}$ & 275 & 14 & 142 & 13 & 137 & 8 \\
2PC;$\{p_2,p_3,p_4\}$ & 63991 & TO & 17894 & 10169 & 17927 & 8822 \\
2PC;$\{p_3,p_4,p_8\}$ & 65845 & TO & 23515 & TO & 23843 & 12547 \\
ABP;$\{s_1,s_2,s_3\}$ & 16797 & 851 & 2987 & 718 & 2773 & 198 \\
ABP;$\{s_2,s_3,s_4\}$ & 53138 & 3631 & 9879 & 2631 & 8088 & 694 \\
ABP;$\{s_3,s_4,s_5\}$ & 12407 & 616 & 2288 & 515 & 2205 & 145 \\
ABP;$\{s_4,s_5,s_6\}$ & 11202 & 564 & 2086 & 489 & 1997 & 140 \\

 2PC;$\{p_1,p_2,p_3,p_4\}$ & 60616 & TO & 9524 & TO & 11259 & TO \\
 ABP;$\{s_1,s_2,s_3,s_4\}$ & 120050 & TO & 11993 & TO & 12988 & 14626 \\
 ABP;$\{s_3,s_4,s_5,s_6\}$ & 43215 & 2577 & 2323 & 2173 & 1967 & 333 \\

\hline
\end{tabular}
\caption{1-Process Synthesis, Unrealizable 1-Completion.
The timeout is checked at the beginning of each iteration, so a program may terminate in just over 4 hours without timing out if the timeout occurs in the middle of the last iteration, as is the case for the results of the experiment where $A = \{s_1, s_2, s_3, s_4\}$ in the perm. opt. column.
}
\label{non-zero-table}
\end{table}